\documentclass{article}
\usepackage[utf8]{inputenc}
\usepackage{amsmath}
\usepackage{amssymb}
\usepackage{amsfonts}
\usepackage{amsthm}
\usepackage{esvect}
\usepackage{cite}
\usepackage{physics}
\usepackage{mathrsfs}
\usepackage{authblk}
\usepackage{geometry}
\usepackage{abstract}
\usepackage{xcolor}
\usepackage{enumerate}
\newcounter{NN}
\newtheorem{prop}[NN]{Proposition}
\newtheorem{ques}[NN]{Question}
\newtheorem{lem}[NN]{Lemma}

\newtheorem{cor}[NN]{Corollary}
\newtheorem{defi}[NN]{Definition}
\newtheorem{exam}[NN]{Example}

\def\is{\iota}
\def\fu{\mathbb{F}} 
\def\F{{\cal F}}
\def\G{{\cal G}}

\def\C{{\cal C}}

\def\cA{{\cal A}} 
\def\cB{{\cal F}} 
 
\def\bbF{\mathbb{F}} 
\def\bbR{\mathbb{R}}
\def\bbC{\mathbb{C}}
\def\ham{{\cal X}} 

\newcommand{\Rk}{\mathop{\rm Rk}\nolimits}
\newcommand{\FI}{\mathop{\rm FI}\nolimits}

\newcommand{\PB}{\left\{\cdot\,,\cdot\right\}}

\newcommand{\Pb}[1]{\left\{\cdot\,,#1\right\}}
\renewcommand{\pb}[1]{\left\{#1\right\}}

\title{\bf Subalgebras of integrals, commutants, and superintegrable deformations of Lotka-Volterra systems}

\author{\bf
Ian Marquette$^{a}$\footnote{i.marquette@uq.edu.au},
Peter H. van der Kamp$^{b}$\footnote{p.vanderkamp@latrobe.edu.au},
G.R.W. Quispel$^{b}$\footnote{r.quispel@latrobe.edu.au}
}

\medskip

\affil{$^{a}$ \quad Department of Mathematical and Physical Sciences, La Trobe University,\\
Bendigo VIC 3552, Australia \\
$^{b}$ \quad Department of Mathematics, La Trobe University,\\
Bundoora VIC 3086, Australia}

\usepackage[colorlinks=true,urlcolor=black]{hyperref}

\begin{document}

\maketitle
\begin{abstract}
We consider the Lie-algebraic notion of commutant in the setting of Poisson algebra. This provides a framework for deforming Hamiltonian differential equations. By taking a subalgebra of the algebra of integrals, and considering the set of functions that Poisson commute with that subalgebra, the
Hamiltonian can be deformed, while retaining integrability. We deform Liouville integrable and superintegrable Lotka-Volterra systems studied in \cite{KKQTV}. We present different explicit constructions considering Abelian and non-Abelian subalgebras of integrals. We obtain superintegrable systems for specific dimensions, and in arbitrary dimension. Polynomial systems are deformed to rational systems, some of which have non-rational integrals. Superintegrability seems to be preserved in this approach.
\end{abstract}

\section{Introduction} 
Liouville integrability and superintegrability have been the object of an extensive literature, which includes various generalisations of Kepler-Coulomb and harmonic oscillator Hamiltonians such as Turbiner-Tremblay-Winternitz and Smorodinsky-Winternitz potentials \cite{win67,eva90,mil13}.
A wide range of superintegrable systems exists beyond integrable Hamiltonians with scalar interaction, including e.g. magnetic fields and spin.

Various properties are known such as exact and quasi-exact solvability in classical and quantum mechanics, the existence of algebra of integrals beyond finite-dimensional Lie algebras and connection with orthogonal polynomials and special functions. In that context, the canonical Poisson bracket is used among generalised variables and generalised momenta. The related integrals are functions of the coordinates, and polynomial in the momenta, whose degree determines the order of the integral. For superintegrable systems with second order integrals the related quadratic (Lie/Poisson) algebra, their Casimir and representations play an important role, in particular to obtain the spectrum of the quantised version \cite{das01, kal02,lia18,mar23}.

This paper aims to build on a series of works on the notion of the commutant of Lie and Poisson-Lie algebras \cite{cam23, cam23b, cam24, cam24b, mar25, mar26a, mar26b}. In those works, since the Poisson-Lie algebras have connection with root systems, grading, and different structural properties of the related Lie algebras, this allows one to obtain an explicit formula for the polynomials in the commutative variables of the Poisson-Lie setting or the non-commutative variable in the related Lie setting. It was demonstrated how the commutant can be constructed explicitly in the case of positive roots subalgebras \cite{cam24} and for Cartan subalgebras of $sl(n)$ and for all simple Lie algebras via indecomposable polynomials \cite{cam23,cam24b}. The latter relates to the well-known Racah algebra, and generalisations, that appear in the mathematical physics literature \cite{GVVZ,CGAV,lat21}. In the case of singular subalgebra embedding that appears in nuclear physics, such as the Elliott chains \cite{cam23b}, the commutant admits more complicated polynomials, which relate to the set of missing labeling operators.
In this context, the Poisson-Lie analog is used primarily as a tool to obtain the leading terms in the Lie setting/quantum setting, which can be obtained by means of the symmetrisation map (and appropriate correction terms from the quantisation \cite{cam23c}). In this paper, the approach is different; we aim to extend the construction of a commutant as a generalisation of Casimir invariants, and use it to deform Liouville and superintegrable systems of nonlinear differential equations.

A Liouville integrable system consists of a Hamiltonian with $m$ degrees of freedom (on a $2m$-dimensional phase space) requiring $m$ mutually commuting integrals of motion (including the Hamiltonian). We note that different definitions of classical and quantum superintegrability are in use. In one setting, a superintegrable system is a Liouville integrable system with $0<k<m$ additional integrals (Poisson-commuting with the Hamiltonian). In another (non-Hamiltonian) setting, a system is superintegrable if there exist $n-1$ integrals for an $n$-dimensional system. In this paper, we consider  $n$-dimensional Hamiltonian systems. Such a system is \emph{superintegrable of rank $r$} if there are $s$ functionally independent integrals, $r$ of which commute with all of them, and $r + s = n$.

Lotka-Volterra (LV) systems are among the best known dynamical systems for their application in the context of biology and description of nonlinear phenomena. It was only discovered in recent years that large families of $n$-dimensional Liouville integrable and superintegrable LV-systems exist. Both notions of superintegrability apply. There are Liouville integrable systems which are also superintegrable \cite{Bog,KKQTV,CHK,KMQ,QPW} and there are non-Hamiltonian superintegrable (or nonholonomically integrable) LV-families \cite{treeL,KMMQ,treeP,Hyper,QPW}. In this paper, we reexamine the Liouville superintegrable LV-systems, with $n\in\mathbb{N}$,
\begin{equation} \label{silv}
\dot{x}_i=x_i\left(-\sum_{j=1}^{i-1} x_j + \sum_{j=i+1}^{n} x_j  \right),\quad i=1,\ldots,n,
\end{equation}
whose Poisson algebra of integrals of motion was obtained in 2014 \cite{KKQTV}. We demonstrate that the Jacobi identity is satisfied not only by the quadratic Poisson algebra of the initial coordinates, but also by the bracket induced by the algebra of integrals, taking the integrals as new coordinates. In the even-dimensional case, we obtain another quadratic Poisson algebra and in the odd-dimensional case we obtain a Poisson algebra with rational expressions. In other contexts of superintegrable systems, to obtain a polynomial algebra one needs additional relations, called closure relations. This is the case for the Racah algebras and related embeddings \cite{lat21}. 

The paper is organised as follows. In Section \ref{sec2}, we provide the definition of a Liouville integrable system (with Casimirs). We prove a basic lemma that tells us that the algebra of integrals defines a Poisson bracket on the space of integrals. We define the notion of commutant relative to a Poisson subalgebra. And we provide a method to prove the completeness of a set of integrals. In Section \ref{sec3}, we present the integrals for LV-system \eqref{silv} and the Poisson algebra they satisfy. In Section \ref{sec4}, we outline how to deform a Hamiltonian system using the generators of a subalgebra of integrals. The deformed system admits the commutants as integrals. As an example, we deform the LV-system \eqref{silv} with a non-Abelian subalgebra of integrals, and obtain superintegrable systems (of different ranks) in both the new variables and in the initial ones.
In Section \ref{sec5}, we consider a discrete family of Abelian subalgebras, in arbitrary dimensions ($n$ odd and $n$ even). In Section \ref{sec6}, we consider 1-dimensional algebras of integrals, where the generator is a linear combination. Here we obtain superintegrable systems with nonalgebraic integrals. In Section \ref{sec7}, we show how to iterate the construction, that is, how to deform deformed systems.

We explicitly provide commutants, and prove completeness, for a variety of subalgebras of integrals, the Poisson algebras they satisfy, as well as the deformed systems. The results obtained point out the wide applicability of our method. 

\section{Poisson algebras, commutants and superintegrable systems} \label{sec2}
In this section we recall the basic definitions of Liouville and superintegrable systems, which we adapt to the case of rational functions. Since the phenomena which we wish to describe in this paper are algebraic and very generally valid, we can take as our base field any field $\bbF$ of characteristic zero, keeping in mind the examples of $\bbR$ and $\bbC$; also, we use Poisson brackets, which are more general and more algebraic than symplectic structures.

\smallskip

Our basic algebra will be $\cA_n=\bbF(x_1,\ldots,x_n)$, the field of fractions over $\bbF$ in $n$ (commuting) indeterminates. Geometrically, one may think of $\cA_n$ as the field of rational functions on $\bbF^n$ with values in $\bbF$. Such a function can be evaluated at a generic point of $\bbF^n$ and is smooth when $\bbF=\bbR$, and is holomorphic when $\bbF=\bbC$, on an open dense subset of $\bbF^n$. A \emph{Poisson bracket} on $\cA$ is a Lie bracket $\{\cdot,\cdot\}:\cA\times\cA\rightarrow\cA$ which satisfies the Leibniz rule in each argument, i.e., is a derivation in each argument. It follows that the Poisson bracket of $f,g\in\cA_n$ can be computed from the Poisson brackets $\pb{x_i,x_j}$, with $1\leqslant i<j\leqslant n$,
\begin{equation}\label{eq:pb}
  \pb{f,g}=\sum_{i,j=1}^n\pb{x_i,x_j}\frac{\partial f}{\partial x_i}\frac{\partial g}{\partial x_j}\;.
\end{equation}
It makes $\cA_n$ into a \emph{Poisson algebra} $(\cA_n,\PB)$. It follows from \eqref{eq:pb}, which is valid for any biderivation on $\cA_n$, that $\PB$ satisfies the Jacobi identity for any triplet of elements of $\cA_n$ as soon as the Jacobi identity is satisfied for all triplets $(x_i,x_j,x_k)$ with $1\leqslant i<j<k\leqslant n$. Thus the matrix $x_{i,j}=\pb{x_i,x_j}$ contains all the information about the Poisson bracket; it is called the \emph{Poisson matrix} of $(\cA_n,\PB)$ and its rank is called the \emph{Poisson rank} of $(\cA_n,\PB)$; the Poisson rank is an even number, smaller than or equal to $n$, denoted $\Rk\PB$. It can be computed by evaluating the Poisson matrix at a generic point of $\bbF^n$. 

\smallskip

Poisson subalgebras will play an important role in this paper. Given a Poisson algebra $\cA$, a subset $\cB\subset\cA$ is a \emph{Poisson subalgebra} when it is both a subalgebra (that is, $\cB\cB\subset\cB$) and a Lie subalgebra (that is, $\pb{\cB,\cB}\subset\cB$). One may then restrict the product and Poisson bracket of $\cA$ to $\cB$, which makes $\cB$ into a Poisson algebra. Suppose now that for a finite subset $\{f_1,\dots,f_k\}$ of $\cA_n$, for any $i<j$, $\pb{f_i,f_j}$ can be written as a rational function of $f_1,\dots,f_k$. Then the subalgebra $\F$ generated by $f_1,\dots,f_k$ is a Poisson subalgebra of $\cA_n$. A particular case is when the functions are functionally independent. Then the subalgebra generated by $f_1,\dots,f_k$ is the field of fractions $\F=\bbF(f_1,\dots,f_k)$. Let us introduce the matrix $f_{i,j}(f)$ such that
\begin{equation} \label{fijx}
\{f_i(x),f_j(x)\}=f_{i,j}(f(x)).
\end{equation}
\begin{lem} \label{LPB}
The matrix $f_{i,j}$ defines a bracket on $\F$, $\{\cdot,\cdot\}_\F:\F\times\F\rightarrow\F$,
\begin{equation} \label{pbf}
\{f_i,f_j\}_\F=f_{i,j}.
\end{equation}
\end{lem}
\begin{proof}
Suppose the Jacobi identity is not satisfied. Then there exist $i<j<k$ such that
\[
R_{i,j,k}(f)=\{f_i,\{f_j,f_k\}_\F\}_\F+\{f_j,\{f_k,f_i\}_\F\}_\F+\{f_k,\{f_i,f_j\}_\F\}_\F\neq 0.
\]
However, as a function of $x_1,\ldots,x_n$ we have $R_{i,j,k}(f(x))=0$. Therefore, $R_{i,j,k}(f)=0$ is a non-trivial relation between $f_1,\ldots, f_m$, in contradiction with $f_1,\ldots, f_m$ being functionally independent.
\end{proof}

\begin{exam} \label{ex1}
   Consider on $\cA_5:=\bbF(x_1,\ldots,x_5)$ the quadratic Poisson bracket defined by
\begin{equation}
\{x_i,x_j\}= x_i x_j,\quad i<j.
\end{equation} 
And consider the functions
\[
H=x_1+x_2+x_3+x_4+x_5,\quad F=\frac{\left(x_{1}+x_{2}+x_{3}\right) x_{5}}{x_{4}},\quad
G=\frac{\left(x_{5}+x_{4}+x_{3}\right) x_{1}}{x_{2}},\quad
C=\frac{x_{1} x_{3} x_{5}}{x_{2} x_{4}}.
\]
One can verify that the only nonvanishing bracket between these functions is
\[
\{F,G\}=-\frac{x_{5} x_{1} \left(x_{4}+x_{5}\right) \left(x_{1}+x_{2}\right)}{x_{2} x_{4}}.
\]
On the space $\F=\fu(H,F,G,C)$ we have
\[
\{F,G\}_\F=CH-FG,
\]
which, together with $\{H,\dot\}=\{C,\cdot\}=0$, defines a bracket on $\F$.
\end{exam}

Let us recall how a Poisson bracket on $\cA$ is used to construct Hamiltonian systems.  Fixing one of the arguments of the Poisson bracket to $h\in\cA$, we get a derivation $\ham_h:=\Pb h$ which is called the \emph{Hamiltonian derivation} associated to $h$; in geometrically terms, $\ham_h$ is a (rational) vector field on $\fu^n$.
When $\ham_h=0$, one says that $h$ is a (rational) \emph{Casimir}. When $f,h\in\cA$ have zero Poisson bracket, $\pb{f,h}=0$, they are said to be \emph{in involution}. Then $\ham_hf=0$, that is, $f$ is a first integral of $\ham_h$.  Given any subset $\cB\subset\cA$, the \emph{centraliser} or \emph{commutant} of $\cB$ in $\cA$ is defined by
\begin{equation}
C_{\cA}(\cB):=\left\{f\in\cA\mid\pb{f,g}=0\hbox{ for all } g\in\cB\right\}\;.
\end{equation}
Taking $\cB:=\{h\}$ it is clear that $C_{\cA}(h)$ consists of the first integrals of $\ham_h$. In what follows we will therefore denote $C_{\cA}(h)$ also by $\FI(h)$. For any $h\in\cA$, $\FI(h)$ is a Poisson subalgebra of $\cA$.
We now define the notion of Liouville and superintegrability, cf. \cite[Definition 3.1]{CEP} for the analogues definition in the smooth case (and note that there the systems are called non-commutative integrable).
  
\begin{defi}\label{def:non-com}
Suppose that $\cA_n=\bbF(x_1,\dots,x_n)$ is equipped with a Poisson bracket $\PB$. An $s$-tuplet $(f_1=h,f_2,\dots,f_s)$ of elements of $\cA_n$ is said to be a \emph{superintegrable system of rank $r$} on $(\cA_n,\PB)$ if
\begin{enumerate}
  \item[(1)] $f_1,\dots,f_s$ are algebraically independent;
  \item[(2)] The functions $f_1,\dots,f_r$ are in involution with the functions $f_1,\dots,f_s$;
  \item[(3)] $r+s =\dim M$;
  \item[(4)] The Hamiltonian derivations $\ham_{f_1},\dots,\ham_{f_r}$ are linearly independent at some point of $\bbF^n$.
\end{enumerate}
A single element $h$ of $\cA_n$ is said to be \emph{superintegrable of rank $r$} if there exists a superintegrable system $(f_1,\dots,f_s)$ of rank $r$ on $\cA_n$, with $f_1=h$.
\end{defi}
Notice that $r\leqslant\frac12\Rk\PB$, as a consequence of (4), which prevents any of the $f_1,\ldots,f_r$ to be a Casimir. Superintegrable systems of rank $r=\frac12\Rk\PB$ are more commonly called \emph{Liouville integrable systems}, reserving the terminology superintegrable for the case $r<\frac12\Rk\PB$. On the other extreme, when $r=1$ one speaks of \emph{maximal superintegrability}\footnote{Many authors use the terminology \emph{non-commutative integrable} for what is called here superintegrable and use the terminology \emph{superintegrable} for what we call here maximally superintegrable}. From the geometrical point of view, the $s$ functions $f_1,\dots,f_s$ define a rational map whose generic fibers are $r$-dimensional; when $\bbF=\bbR$ the compact connected components of these fibers are tori, known as the \emph{Liouville tori}. An action-angle theorem for superintegrable systems of rank $r$ was proven in \cite[Theorem 3.6]{CEP}.

\begin{exam}
Continuing Example 1, the Hamiltonian $H$ gives rise to the system of ordinary differential equations (ODE), $\dot{x}=\{x,H\}$, or
\begin{equation} \label{n=5} 
\begin{split}
    \dot{x}_1&=x_1(x_2+x_3+x_4+x_5)\\
    \dot{x}_2&=x_1(-x_1+x_3+x_4+x_5)\\
    \dot{x}_3&=x_1(-x_1-x_2+x_4+x_5)\\
    \dot{x}_4&=x_1(-x_1-x_2-x_3+x_5)\\
    \dot{x}_5&=x_1(-x_1-x_2-x_3-x_4).
\end{split}
\end{equation}
The functions $F,G$ are integrals, and the function $C$ is a Casimir. The system is Liouville integrable, or superintegrable of rank 2, in two different ways, with tuple of integrals
$(H,F,C)$ and with $(H,G,C)$. It is also superintegrable of rank 1 with integrals $(H,F,G,C)$.

On the space $\F=\fu(H,F,G,C)$, the tuples $(F,H,C)$ and $(G,H,C)$ are superintegrable of rank 1. For example, taking $F$ as the Hamiltonian we obtain the ODE
\begin{equation}
\begin{split}
\dot{H}&=0\\
\dot{F}&=0\\
\dot{G}&=FG-CH\\
\dot{C}&=0.
\end{split}
\end{equation}

Technically, one can view the tuple $(H,F,G,C)$ as superintegrable of rank 0, as $H$ has trivial flow. 
\end{exam}

It may be very difficult to give a complete description of the integrability for a given Hamiltonian $h$, even knowing that a Hamiltonian is integrable or not is in general a difficult task. It relies on being able to solve a system of partial differential equations (PDE). As seen in the examples, a given $h\in\cA$ may be superintegrable of different ranks, as one may sometimes choose among the first integrals different sets, where for example one set makes it Liouville integrable and another set makes it maximally superintegrable. In general, superintegrability of some rank $r$ does not imply superintegrability of some lower or higher rank and it is in general very difficult to determine for which ranks from $1$ to $\frac12\Rk\PB$ it is superintegrable, if any. All one can do is exhibit first integrals and check if they can be arranged so as to prove superintegrability of some rank.

\smallskip
We mention two further questions, which were raised by Pol Vanhaecke \cite{Pol} and are natural questions to ask for every system one looks at.
\begin{ques} \label{q1}
For a given set of rational first integrals $h=f_1,f_2,\ldots,f_k\in\cA$ of $h$, is every rational first integral of $h$ a rational function of the given ones? Said differently, is it true that $\FI(h)=\bbF(f_1,\ldots,f_k)$?
\end{ques}
In the above example, since $H\in\cA_5$ is superintegrable of rank 1, there is a rational relation between any first integral $K$ and the given first integrals, but that is not sufficient as the relation may not be linearly solvable for $K$.
\begin{ques} \label{q2}
For a given set of rational functions $h=f_1,f_2,\ldots,f_s\in\cA$ , with the first $r$ of them pairwise commuting, is every rational first integral of $h$ in involution with $f_1,f_2,\ldots,f_r$, a rational function of the given functions? Said differently, is it true that $C_{\cA}(f_1,\ldots,f_r)=\bbF(f_1,\ldots,f_s)$?
\end{ques}
An answer to Question 5 in the case of polynomial Liouville integrable systems was given  in \cite[Chapter II, Prop. 3.7]{Van}. We provide an approach to answer such questions in the rational case. For example, to answer Question \ref{q1}, one can choose a subset of variables, $S\subset V$, such that the transformation from $V$ to $I \cup S$ is a birational map, and then prove that the condition for a function to be an integral is that the function does not depend on $S$.

\begin{exam} \label{exm}
On $\cA_5$, we consider the functions $z=[H,F,C,G,x_5]$. If the map $x_i\rightarrow z_i$, $i=1,\ldots,5$ is birational, i.e. we can express the $x_i$ as rational functions of $z_1,z_2,z_3,z_4,z_5$, then each rational function $K(x)$ is a rational function of the integrals, and $x_5$. The map is birational indeed, we have
\begin{align*}
x_{1} &= 
\frac{\left(z_{3}-z_{5}\right) \left(z_{1} z_{3}-z_{2} z_{4}\right)}{z_{1} z_{3}-z_{1} z_{5}-z_{2} z_{4}-z_{2} z_{5}-z_{3} z_{5}-z_{4} z_{5}}
,\\ x_{2} &= 
-\frac{\left(z_{3}-z_{5}\right) \left(z_{3}+z_{2}\right) \left(z_{1} z_{3}-z_{2} z_{4}\right) z_{5}}{\left(z_{1} z_{3}-z_{1} z_{5}-z_{2} z_{4}-z_{4} z_{5}\right) \left(z_{1} z_{3}-z_{1} z_{5}-z_{2} z_{4}-z_{2} z_{5}-z_{3} z_{5}-z_{4} z_{5}\right)}
,\\ x_{3} &= 
-\frac{\left(z_{3}+z_{2}\right) \left(z_{3}-z_{5}\right) z_{1} z_{5}}{\left(z_{1} z_{3}-z_{1} z_{5}-z_{2} z_{4}-z_{4} z_{5}\right) \left(z_{2}+z_{5}\right)}
,\\ x_{4} &= \frac{z_{5} \left(z_{3}-z_{5}\right)}{z_{2}+z_{5}},\\ x_{5} &= 
z_{5}
\end{align*}
Moreover, the condition for $K(x)$ to be an integral, $\{H, K(x)\} =0$, yields the following PDE
\begin{align*}
 -x_{1} \left(x_{2}+x_{3}+x_{4}+x_{5}\right) K_{x_{1}}+x_{2} \left(x_{1}-x_{3}-x_{4}-x_{5}\right) K_{x_{2}}+&x_{3} \left(x_{1}+x_{2}-x_{4}-x_{5}\right) K_{x_{3}}\\
 +x_{4} \left(x_{1}+x_{2}+x_{3}-x_{5}\right) K_{x_{4}}+&x_{5} \left(x_{1}+x_{2}+x_{3}+x_{4}\right) K_{x_{5}}
=0,   
\end{align*}
and this PDE transforms into
\[
\left(z_{3}-z_{5}\right) z_{5} K_{z_{5}} = 0.
\]
Hence, rational integrals depend rationally on the known integrals only.
\end{exam}

In the next section, we review the first integrals of a $n$-dimensional Lotka-Volterra system presented in \cite{KKQTV}. The equation \eqref{n=5} is the special case taking $n=5$. In arbitrary dimensions, we have not been able to prove the birationality of the change of variables, but we conjecture this to be the case, and we provide explicit formulas in Appendix \ref{secA}.

\section{A family of superintegrable Lotka-Volterra systems} \label{sec3}

In this section, we review from \cite{KKQTV} a family of superintegrable systems on which we will illustrate all the phenomena studied in this paper. 

\smallskip

Consider on $\cA_n:=\bbF(x_1,\ldots,x_n)$ the quadratic Poisson bracket defined by
\begin{equation} \label{PB}
\{x_i,x_j\}= x_i x_j,\quad i<j\;.
\end{equation}
It makes $(\cA_n,\PB)$ into a Poisson algebra. When $n$ is even, $\Rk(\cA_n,\PB)=n$, otherwise $\Rk(\cA_n,\PB)=n-1$. Taking
\begin{equation} \label{Ham}
  H:=\sum_{i=1}^{n} x_i\;,
\end{equation}
the Hamiltonian derivation $\ham_H=\dot{x}=\{x,H\}$ is completely described by
\begin{equation} \label{LV}
\dot{x}_i=x_i\left(-\sum_{j=1}^{i-1} x_j + \sum_{j=i+1}^{n} x_j  \right),\quad i=1,\ldots,n\;.
\end{equation}
This particular Lotka-Volterra system appeared as a subsystem of the quadratic vector fields of the $\Upsilon$-systems that were associated with multi-sums of products in \cite{KKQTV}; it was shown in \cite{DEKV,QPW} that it is a reduction of the Bogoyavlenskij system \cite{Bog} %
\begin{equation} \label{bog}
\dot{x}_i=x_i\,\big(\sum_{j=1}^{n-1} x_{i+j} - \sum_{j=1}^{n-1} x_{i-j}\big)\;,
\end{equation}
which generalises the well-known Volterra lattice. In \cite{KKQTV} it was shown that the Hamiltonian in \eqref{Ham} is both Liouville integrable and superintegrable, as we now recall. A first set of $m:=\lfloor(n+1)/2\rfloor$ first integrals of \eqref{LV} is given by   
\[
F_k:= \sum_{\ell=1}^{2k-1} x_\ell \prod_{s=k}^{(n-1)/2} \frac{x_{2s+1}}{x_{2s}},\quad k=1,\ldots,m
\]
for $n=2m-1$ odd, and 
\[
F_k:= \sum_{\ell=1}^{2k} x_\ell \prod_{s=k}^{n/2-1} \frac{x_{2s+2}}{x_{2s+1}},\quad k=1,\ldots,m
\]
for $n=2m$ even. Another set of integrals is obtained using the involutive algebra homomorphism $\imath:\cA_n\to\cA_n$, defined by
\begin{equation} \label{istar}
\imath(x_i):=x_{n+1-i},\qquad i=1,\dots,n.
\end{equation}
This other set is given by $G_k:=\imath(F_k)$, where $k=1,\dots,m$. We have $F_m=G_m=H$, and, when $n$ is odd, $F_1=G_1=C$, where $C$ is the following Casimir of the Poisson bracket~\eqref{PB}:
\[
C:=\frac{x_1x_3\cdots x_{n-2}x_{n}}{x_2x_4\cdots x_{n-1}}\;,
\] 
and, besides the indicated equalities, all these first integrals are algebraically independent and there are $n-1$ of them, \cite[Appendix A]{KKQTV}. The integrability of $H$ that was proven in \cite{KKQTV} can be summarised as follows.
\begin{prop} \label{psum}
\begin{itemize}
\item $H$ is Liouville integrable with involutive first integrals $F_1,F_2,\dots,F_m=H$.
\item $H$ is Liouville integrable with involutive first integrals $G_1,G_2,\dots,G_m=H$.
\item When $n=2m$ is even, $H$ is maximally superintegrable with first integrals $F_1,G_1,F_2,G_2,\dots,F_{m-1},G_{m-1},$ $F_m=G_m=H$.
\item When $n=2m-1$ is odd, $H$ is maximally superintegrable with first integrals $F_1=G_1,F_2,G_2,\dots,F_{m-1},$ $G_{m-1},F_m=G_m=H$.
\end{itemize}
\end{prop}
The given first integrals generate a Poisson subalgebra of $(\cA_n,\PB)$ which we denote by $\cB$, so that when $n=2m$ is even,
\begin{equation} \label{Fe}
 \cB=\bbF(H,F_1,G_1,\dots,F_{m-1},G_{m-1}),
\end{equation}
while when $n=2m-1$ is odd, \begin{equation} \label{Fo}
\cB=\bbF(H,C,F_2,G_2,\dots,F_{m-1},G_{m-1}).
\end{equation}
It is clear that $\cB\subset\FI(H)$, possibly with equality, cf. Question \ref{q1}.  

\smallskip
In order to describe the Poisson structure on $\cB$, inherited from the Poisson bracket on $\cA_n$, let $1\leqslant i,j<m$ and denote
\begin{equation} \label{kij}
\kappa:=\kappa(i,j):=i+j-m-1\;.
\end{equation}
It was shown that
\begin{equation} \label{PBC}
\{H,F_i\}=0,\quad \{H,G_i\}=0,\quad
\{F_i,F_j\}=0,\quad \{G_i,G_j\}=0,
\end{equation}
and, for $n$ even, 
\begin{equation} \label{PBE}
\{F_i,G_j\}=
\begin{cases} -F_i G_j & \kappa < 0 \\
- F_{m-j} G_{m-i} &  \kappa \geqslant 0,
\end{cases}
\end{equation}
and, for $n$ odd,
\begin{equation} \label{PBO}
\{F_i,G_j\}=
\begin{cases} 0 & \kappa < 0 \\
-\dfrac{\prod_{l=0}^{\kappa} (F_{i-\kappa +l} G_{j-l} -CH )}{\prod_{l=1}^{\kappa} (F_{i-\kappa +l-1} G_{j-l} + CH )} &  \kappa \geqslant 0.  \end{cases}
\end{equation}


Since all elements of $\cB_n$ are first integrals of $H$, $\pb{\cB_n,H}=0$, so that $H$ is a Casimir of $(\cB_n,\PB)$. It follows that $\Rk(\cB_n,\PB)=\Rk(\cA_n,\PB)-2$. Apart from Prop. \ref{psum}, the Poisson structure on $\cB$ also implies that the odd $n$-dimensional system is superintegrable of rank $R$, for $R=1,2,\ldots,(n-1)/2$. 
\begin{prop}
For odd $n=2m-1$, we have with $R=1,\ldots,m-1$, the tuplets
\[
(F_m,F_2,\ldots,F_{m-1},G_2,\ldots,G_{m-R},F_1),\quad
(G_m,G_2,\ldots,G_{m-1},F_2,\ldots,F_{m-R},G_1)
\]
are superintegrable systems of rank $R$.
\end{prop}

The considerations in the previous section imply that one can view the symbols $H,C,F_i,G_i$ as independent variables (instead of functions) and regard the formulas \eqref{PBC} with \eqref{PBE} or \eqref{PBO} as the {\em definition} of a Poisson bracket of $\F$.
\begin{prop} \label{PAs} Depending on the parity of $n$ we have:
\begin{enumerate}[(i)]
    \item For $n\equiv 0$ the brackets \eqref{PBC} and \eqref{PBE} define a Poisson bracket on $\cB=\bbF(H,F_1,G_1,\dots,F_{m-1},G_{m-1})$ (where $H$ is a constant/Casimir).
    \item For $n\equiv 1$ the brackets \eqref{PBC} and \eqref{PBO} define a Poisson bracket on $\cB=\bbF(H,C,F_2,G_2,\dots,F_{m-1},G_{m-1})$ (where $C$ and $H$ are constants/Casimirs). 
\end{enumerate} 
\end{prop}
The statement is clear, as the Poisson algebras of integrals inherit their structures as Poisson subalgebras, from the definition of the functions and the fact that the integrals are independent, see Lemma \ref{LPB}. However, the statement is nontrivial, which we illustrate by providing a proof, in Appendix \ref{secB}, which does not rely on the definitions of the functions $F_i$ and $G_j$, the original bracket \eqref{PB}, Lemma \ref{LPB}, and the results of \cite{KKQTV}.

\smallskip
Recently, homogeneous integrable generalizations of Lotka-Volterra \eqref{LV} were studied in \cite{treeL,treeP}, whereas inhomogeneous integrable generalizations were studied in \cite{KMQ,CHK,QPW}. Using the Poisson algebra of symmetries. We will obtain and study integrable generalizations of \eqref{LV}, by deforming the Hamiltonian using a Poisson subalgebra and the notion of commutant introduced in Section \ref{sec2}.

\section{Hamiltonian deformation using a subalgebra of integrals} \label{sec4}
We say that $\G=\fu(g_1,\ldots,g_l)$ is a \emph{subalgebra of integrals} if
\[
\{g_i,g_j\}_\F \in \G\subset\F=\FI(f_1)=\fu(f_1,\ldots,f_k).
\]
In the sequel, we will take $g_1=f_1$ ($=H$, the Hamiltonian) and assume that the functions $g_i$, $i=1,\ldots,l$, are functionally independent, so that $\G$ is a Poisson algebra, cf. Lemma \ref{LPB}. The \emph{commutant} of the  subalgebra of integrals $\G \subset \F$ is the set of partial Casimirs
$\C=C_\F(\G)$. We will aim to find rational (and functionally independent) functions $c_i$, so that $\C=\fu(c_1,\ldots,c_d)$ with $c_j\in\F$. As $H\in\C$, we will set $c_1=H$. Note that the functions $c_i$ do not have to Poisson commute with each other. We will also call each $c_i\in\C$ a commutant. The condition that $c\in\C$ is a commutant of $\G \subset \F$ is the system of partial differential equations 
\begin{equation} \label{compdes}
\{c, g_i\}= \sum_{j=1}^{k} \{f_j,g_i\} \frac{ \partial c}{\partial f_j}=0,\quad i=1,\ldots,l.   
\end{equation}
This provides the possibility for $c$ to take the form of polynomials or, more general, rational and even transcendental functions (as we will see in Section 6). 

To an algebra of integrals $\F=\fu(f_1,\ldots,f_m)$, with $H=f_1$, corresponds the Hamiltonian system
\begin{equation} \label{HSDET}
\dot f_i = \{f_i,H\}_\F = 0,\quad i=1,\ldots,m.
\end{equation}
That is, in the new $f$-coordinates the Hamiltonian system is trivialised; readily solved by $f_i=constant$ for all $i$.
We propose to deform the system \eqref{HSDET}, and the related system on $\cA_n$, by using a subalgebra  of integrals $G\subset \F$ and its commutant $\C$. We take  $\widehat{H}= \widehat{H}(g_1,\ldots,g_l)$ as the deformed Hamiltonian, and then the corresponding system admits the commutants $c_1,\ldots,c_d$ as its integrals of motion. The original Hamiltonian $H$ and any Casimirs are included in the integrals of these deformed systems.

Let us present more details on the approach we take towards solving the commutant constraints (\ref{compdes}), which in general is a difficult problem. Here, we rely on solving the systems of PDEs using the method of characteristics implemented in symbolic computation software  \cite{Map}. This provides sets of solutions 
\[
c_i = c_i(f_1,\ldots,f_m),\quad i=1,\ldots,d.
\]
We then use the solutions obtained to perform a change of variables
\[
z_i = c_i(f_1,\ldots,f_m) \quad i=1,\ldots,d.
\]
The other variables $z_i$ for $d+1,\ldots,m$ are chosen among the $f_i$ that are not part of the solution obtained already via solving the system of PDEs. If the system of PDEs (\ref{compdes}), in the new variables, takes the form 
\[
m(z_1,\ldots,z_m) \frac{\partial K}{\partial z_s} =0,
\quad s=d+1,\ldots,m,
\]
this implies that the general solution only depends on the already obtained integrals, cf. Example \ref{exm}.

This approach will be implemented in examples of commutants associated with Abelian and non-Abelian subalgebras of integrals. We now choose low-dimensional non-Abelian subalgebras to illustrate the procedure, deforming both the odd- and the even-dimensional system \eqref{LV}.

\subsection{Deformation of the odd $n=2m-1$ dimensional LV system \eqref{LV}} \label{d1}
Let $m>1$ be an integer. We consider the algebra of integrals $\F=\fu(H,F_2,\ldots,F_{m-1},G_2,\ldots,G_{m-1},C)$, where the nontrivial brackets are given by \eqref{PBO}.
We define the subalgebra of integrals
\begin{equation} \label{G1}
\G = \fu(H,F_{2},G_{m-1},C),
\end{equation}
which is non-Abelian as $\{F_{2},G_{m-1}\}_\F=CH-F_{2}G_{m-1}$. The constraints on a commutant, $L$, relative to this subalgebra, is given by the set of equations $\{F_2,L\}_\F=0$,  $\{G_{m-1},L\}_\F=0$,
which amounts to the following system of PDEs
\begin{equation}
\sum_{j=2}^{m-1} \{F_2,G_j\}_\F L_{G_j} =0, \quad \sum_{j=2}^{m-1} \{G_{m-1},F_j\}_\F L_{F_j} =0,  
\end{equation}
where subscripts denote derivation, i.e., $L_{F_j}=\partial L/\partial F_j$. 
Solving the system yields the following solutions
\[  G_{i},\quad i=2,\ldots,m-2, \]
and the rational functions
\begin{equation}
L_i=\frac{F_{i} G_{m-i+1}-C H }{F_{i-1} G_{m-i+1}+CH},\quad i=3,\ldots,m-1.
\end{equation}
\begin{prop}
The commutant of $\cal{G}$ is
$\C=\fu(H,C,G_2,\ldots,G_{m-2},L_3,\ldots,L_{m-1})$.
\end{prop}
We have $\Rk\PB=2m-4$, as there are two Casimirs, $H$ and $C$. The (linearly) deformed Hamiltonian $\widehat{H}= H + \alpha_1 F_2 + \alpha_2 G_{m-1}$ is Liouville integrable with $m-3$ integrals in involution of type $G_i$, as well as $m-3$ integrals in involution of type $L_i$. It is also superintegrable of rank 1, as there are in total $2m-3$ independent integrals,
\[
I=(\widehat{H},G_2,\ldots,G_{r-2},L_3,\ldots,L_{r-1},H,C),
\]
including the deformed Hamiltonian.

\begin{prop} \label{nsa}
Define ${\cal I}=\bbF(I)$. The new algebra of integrals is given by
\[
\{G_i,G_j\}_{\cal I}=0,\quad
\{L_{i+1},L_{j+1}\}_{\cal I}=0,\quad 
\{G_i,L_{j+1}\}_{\cal I}= \delta_{i+j+1,m+1} G_i L_{j+1},
\]
for all $1<i,j<m-1$, with Casimirs $\widehat{H},H,C$.
\end{prop}
\begin{cor}
    Proposition \ref{nsa} implies that the tuple of functions $(I_1,I_2,\ldots,I_{2m-4-R},H,C)$ is superintegrable of rank $R$ for each $R=1,2,\ldots,m-2$.
\end{cor}

\noindent
Let us consider the case $n=7$ ($m=4$) in some more detail.
\begin{exam}
For $n=7$ we have $\F=\fu(H,F_2,F_3,G_2,G_3,C)$. Nonzero brackets are given by 
\[
\{F_2,G_3\}_\F=C H-F_2 G_3, \quad 
\{F_3,G_2\}_\F= C H-F_3 G_2, \quad
\{ F_3, G_3\}_\F= \frac{(F_3 G_2 - C H)(C H-F_2 G_3)}{C H + F_2 G_2}.
\]
We observe that $\G=\fu(H,F_2,G_3,C)$ is a subalgebra of integrals of $\F$. A partial Casimir $K\in \F$ should satisfy:
\begin{equation} \label{PCPDE}
\begin{split}
\{K,F_2\}_\F&=(F_2 G_3 - C H)K_{G_3}=0,\\
\{K,G_3\}_\F&= (C H - F_2 G_3) K_{F_2} +\frac{(F_3 G_2 - C H)(C H - F_2 G_3)}{C H + F_2 G_2} K_{F_3}  =0.    
\end{split}
\end{equation}
Using the method of characteristics, implemented in Maple \cite{Map}, we find two functionally independent functions
\[
G_2,\quad L_3=\frac{F_3 G_2-C H }{C H + F_2 G_2}.
\]
We now prove that the system of PDEs does not admit other integrals than $(H,C,G_2,L_3)$. Consider the birational change of variables
\[
z_1  = L_3, \quad z_2=F_3, \quad z_3 = G_2, \quad z_4 = G_3,
\]
where $C,H$ are constants. In terms of these $z$-variables, the PDEs \eqref{PCPDE} read
\[ \frac{ \left(C H \left(z_4+z_1 \left(z_3+z_4\right)\right)-z_1 z_3 z_4\right)}{z_1 z_3} K_{z_4}=0,\quad \frac{ \left(C H \left(z_4+z_1 \left(z_3+z_4\right)\right)-z_1 z_3 z_4\right)}{z_1 z_3} K_{z_2}=0.
\]
This means the general solution is a function of the known solutions only and no other commutants exist, we have $\C=C_\F(\G)=\bbF(H,C,G_2,L_3)$. We will use $\G$ to deform the system \eqref{HSDET}, that is $
\dot{F_2}=\dot{F_3}=\dot{G_2}=\dot{G_3}=0$,
as well as the system $\dot{x}=\{x,H\}$ with $H=\sum_{i=1}^7 x_i$, see eq. \eqref{LV}. With deformed Hamiltonian
$\widehat{H} = H + \alpha_1 F_2 +\alpha_2 G_3$ and variables $X=(F_2,F_3,G_2,G_3)$ we obtain
\begin{equation}\label{4ds2}
\dot{X}=\{X,\widehat{H}\}_\F=\begin{pmatrix}
\alpha_{2} \left(C H -F_{2} G_{3}\right)\\[1mm]
-\dfrac{\beta_{2} \left(C H -F_{2} G_{3}\right) \left(C H -F_{3} G_{2}\right)}{C H +F_{2} G_{2}}\\[1mm]
0\\[1mm]
-\alpha_{1} \left(C H -F_{2} G_{3}\right),
\end{pmatrix}
\end{equation}
with integrals given by $\widehat{H},G_2,L_3$. Hence, this system is Liouville integrable (superintegrable of rank 2) and superintegrable of rank 1. Enlarging the space with the Casimirs, gives us a 6-dimensional system that is Liouville integrable (superintegrable of rank 2) and superintegrable of rank 1. In terms of the initial coordinates, the system $\dot x=\{x,\widehat{H}\}$ with Hamiltonian
\begin{equation} \label{tHt}
\widehat{H}= \sum_{i=1}^7 x_i + \frac{\alpha_{1} \left(x_{1}+x_{2}+x_{3}\right) x_{5} x_{7}}{x_{4} x_{6}}+\frac{\alpha_{2} \left(x_{7}+x_{6}+x_{5}+x_{4}+x_{3}\right) x_{1}}{x_{2}},
\end{equation}
reads, with anti-symmetric matrix $A$ given by $A_{i,j}=1$, $i<j$, 
\begin{align*}
\dot x_1&=x_{1} \left(\sum_{j=1}^7 A_{1,j}x_j+\beta_{1}\frac{x_{7} x_{5} }{x_{4} x_{6}}\left(x_{2}+x_{3}\right) \right)\\
\dot x_2&=x_{2} \left(\sum_{j=1}^7 A_{2,j}x_j+\beta_{1}\frac{x_{7} x_{5} }{x_{4} x_{6}}\left(-x_{1}+x_{3}\right) \right)\\
\dot x_3&=x_{3} \left(\sum_{j=1}^7 A_{3,j}x_j+\beta_{1}\frac{x_{5} x_{7}}{x_{4} x_{6}} \left(-x_{1}-x_{2}\right) + \beta_{2}\frac{x_{1}}{x_{2}} \left(x_{4}+x_{5}+x_{6}+x_{7}\right)\right)\\
\dot x_4&=x_{4} \left(\sum_{j=1}^7 A_{4,j}x_j+\beta_{2}\frac{ x_{1}}{x_{2}} \left(-x_{3}+x_{5}+x_{6}+x_{7}\right)\right)\\
\dot x_5&=x_{5} \left(\sum_{j=1}^7 A_{5,j}x_j+\beta_{2}\frac{x_{1}}{x_{2}} \left(-x_{3}-x_{4}+x_{6}+x_{7}\right)\right)\\
\dot x_6&=x_{6} \left(\sum_{j=1}^7 A_{6,j}x_j
+\beta_{2}\frac{x_{1}}{x_{2}} \left(-x_{3}-x_{4}-x_{5}+x_{7}\right)\right)\\
\dot x_7&=x_{7} \left(\sum_{j=1}^7 A_{7,j}x_j
+\beta_{2}\frac{x_{1}}{x_{2}} \left(-x_{3}-x_{4}-x_{5}-x_{6}\right)\right),
\end{align*}
and it possesses the following five integrals: $\widehat{H}$ defined by \eqref{tHt}, and
\[
H=\sum_{i=1}^7 x_i,\quad
\frac{1}{G_2L_3}=\frac{\left(x_{4}+x_{5}+x_{6}+x_{7}\right) x_{5}}{x_{4} \left(x_{6}+x_{7}\right)},\quad
\frac{L_3}{C}=\frac{\left(x_{7}+x_{6}+x_{5}\right) x_{6}}{x_{5} x_{7}},\quad
C=\frac{x_{1} x_{3} x_{5} x_{7}}{x_{2} x_{4} x_{6}}.
\]
It is Liouville integrable (superintegrable of rank 3), as well as superintegrable of rank 2. (As the dimension has gone up, it is no longer superintegrable of rank 1.) 
\end{exam}
The new integrals yield a new algebra of integrals, see Proposition \ref{nsa}, of which one can take a Poisson subalgebra and repeat the process. 
 
\subsection{Deformation of the even $n=2m$ dimensional system \eqref{LV}} 
For $1<m\in\mathbb{N}$, we consider the space $\F=\bbF(X)$ with $X=(H,F_1,\ldots,F_{m-1},G_1,\ldots,G_{m-1})$, on which the non-trivial brackets are given by \eqref{PBE}. We define the subalgebra of integrals
\[
\G=\bbF(H,F_{1},G_{m-1}),
\]
which is non-Abelian, as $\{F_{1},G_{m-1}\}_\F=-F_{1}G_{m-1}$. A function $J$ is a commutant of $\G$ if 
\begin{equation} \label{Jpde}
\begin{split}
    \{F_1,J\}_\F &= \sum_{j=1}^{m-1} \{F_1,G_j\}_\F J_{G_j} =0 \\
    \{G_{m-1},J\}_\F &= \sum_{j=1}^{m-1} \{G_{r-1},F_j\}_\F J_{F_j} =0.
\end{split}
\end{equation}

\begin{prop}
The system of PDEs \eqref{Jpde} is satisfied by
\begin{equation}
\widetilde{F}_i = F_{m-i} - F_1  \frac{ G_i}{G_{m-1}},\quad
\widetilde{G}_i = \frac{G_i}{G_{m-1}},\quad  i=1,\ldots,m-2
\end{equation}
The new algebra of integrals is given by
\[
\{\widetilde{G}_i,\widetilde{F}_j\}_{\widetilde{\F}}= \begin{cases}
    \widetilde{G}_j \widetilde{F}_i,\quad  i \geq j,\\
    \widetilde{G}_i \widetilde{F}_j,\quad  i < j.
\end{cases}
\]
\end{prop}

\begin{exam}
For $n=8$ we have $\F=\bbF(H,F_1,F_2,F_3,G_1,G_2,G_3)$. The nonzero brackets are given by 
\begin{align*}
\{F_1,G_1\}_\F&=-F_1 G_1,\quad \{F_1,G_2\}_\F=-F_1 G_2,\quad \{F_1,G_3\}_\F=-F_1 G_3 \\
\{F_2,G_1\}_\F&=-F_2 G_1,\quad \{F_2,G_2\}_\F=-F_2 G_2 ,\quad \{F_2, G_3\}_\F=-F_1 G_2 \\
\{F_3,G_1\}_\F&=-F_3 G_1,\quad \{F_3,G_2\}_\F =- F_2 G_1 ,\quad \{F_3,G_3\}_\F=-F_1 G_1.
\end{align*}
We observe that $\G=\fu(H,F_1,G_3)$ is subalgebra of integrals of $\F$. A partial Casimir $K\in \F$ should satisfy:
\begin{equation} \label{pdese}
\begin{split}
\{F_1,K\}_\F&= -F_1\left(G_1 K_{G_1} + G_2 K_{G_2} + G_3 K_{G_3}  \right) =0,\\
\{G_3,K\}_\F&= F_1\left( G_3 K_{F_1} + G_2 K_{F_2} + G_1 K_{F_3}  \right)  =0.
\end{split}
\end{equation}
Using the method of characteristics, implemented in Maple \cite{Map}, we find four functionally independent functions
\[
\widetilde{F}_1 = F_3-F_1 \frac{G_1}{G_3},\quad \widetilde{F}_2 = F_2 -F_1 \frac{G_2}{G_3},\quad \widetilde{G}_1= \frac{G_1}{G_3},\quad \widetilde{G}_2= \frac{G_2}{G_3}.
\]
We will now prove the system does not admit other integrals than $(H,\widetilde{F}_1, \widetilde{F}_2, \widetilde{G}_1,\widetilde{G_2})$. Using the birational change of variables
\[ 
z_1 = F_1 , z_2 = \frac{(-F_1 G_1 + F_3 G_3)}{G_3} , z_3 = \frac{(-F_1 G_2 +  F_2 G_3)}{G_3}, z_4 = \frac{G_1}{G_3}, z_5 = \frac{G_2}{G_3} ,  z_6 = G_3,
\]
whose inverse is given by
\[
F_{1} = z_{1}, F_{2} = z_{1} z_{5}+z_{3}, F_{3} = z_{1} z_{4}+z_{2}, G_{1} = z_{4} z_{6}, G_{2} = z_{5} z_{6}, G_{3} = z_{6},
\]
the PDEs \eqref{pdese} reduce to $-z_1 z_6 K_{z_6} = z_1 z_6  K_{z_1} = 0$, which means $K$ does not depend on the variables $z_1$ and $z_6$. All other variables are integrals. So, the commutant of $\G\subset\F$ is $\C=\fu(\widetilde{F}_1,\widetilde{F}_2, \widetilde{G}_1,\widetilde{G}_2,H)$. We use $\G$ to deform the system \eqref{HSDET}, i.e.
\[
\dot{F_1}=\dot{F_2}=\dot{F_3}=\dot{G_1}=\dot{G_2}=\dot{G_3}=\dot{H}=0,
\]
as well as the system $\dot{x}=\{x,H\}$ with $H=\sum_{i=1}^7 x_i$, cf. \eqref{LV}.
We take a linear combination $\widehat{H}=H+ \alpha_1 F_1 + \alpha_2 G_3$.
This gives rise to the system, for $X=(F_1,F_2,F_3,G_1,G_2,G_3)$,
\begin{equation}\label{4ds8}
\dot{X}=\{X, \widehat{H}\}_\F=\begin{pmatrix}
-\alpha_2 F_1 G_3\\
-\alpha_2 F_1 G_2 \\
-\alpha_2 F_1 G_1 \\
\alpha_1 F_1 G_1 \\
\alpha_1 F_1 G_2 \\
\alpha_1 F_1 G_3
\end{pmatrix},
\end{equation}
which has five integrals, $\widehat{H},\widetilde{F}_1,\widetilde{F}_2,\widetilde{G}_1,\widetilde{G}_2$. It is both Liouville integrable, superintegrable of rank 3 with commuting integrals $(\widehat{H},\widetilde{F}_1,\widetilde{F}_2)$ and $(\widehat{H},\widetilde{G}_1,\widetilde{G}_2)$, as well as superintegrable of rank 1. (And so is the extended 7-dimensional system one obtains by including the Casimir function $H$.) 

Going back to the realisation in terms of the initial $x$-variables, the deformed Hamiltonian is 
\begin{equation} \label{tilH8}
\widehat{H}= \sum_{i=1}^7 x_i +  \alpha_1  \frac{(x_1 + x_2) x_4 x_6 x_8}{x_3 x_5 x_7}  + \alpha_2 \frac{ x_1 (x_3 + x_4 + x_5 + x_6 + x_7 + x_8)}{x_2}
\end{equation}
which yields the ODE,
\begin{equation} \label{even8}
\begin{split}
\dot x_1&=x_{1} \left( \sum_{j=1}^8 A_{1,j} x_j +\alpha_1 \frac{x_{2} x_{4} x_{6} x_{8}}{x_{3} x_{5} x_{7}} \right),\\
\dot x_2&=x_{2} \left( \sum_{j=1}^8 A_{1,j} x_j - \alpha_1 \frac{x_{1} x_{4} x_{6} x_{8}}{x_{3} x_{5} x_{7}}  \right),\\
\dot x_3&=x_{3} \left( \sum_{j=1}^8 A_{1,j} x_j +\alpha_2 \frac{x_1}{x_2}(x_{4}+x_{5}+x_{6}+x_{7}+x_{8}) \right),\\
\dot x_4&=x_{4} \left( \sum_{j=1}^8 A_{1,j} x_j +\alpha_2 \frac{x_1}{x_2}(-x_{3}+x_{5}+x_{6}+x_{7}+x_{8}) \right),\\
\dot x_5&=x_{5} \left( \sum_{j=1}^8 A_{1,j} x_j +\alpha_2 \frac{x_1}{x_2}(-x_{3}-x_{4}+x_{6}+x_{7}+x_{8}) \right),
\end{split}
\end{equation}
\begin{align*}
\dot x_6&=x_{6} \left( \sum_{j=1}^8 A_{1,j} x_j +\alpha_2 \frac{x_1}{x_2}(-x_{3}-x_{4}-x_{5}+x_{7}+x_{8}) \right),\\
\dot x_7&=x_{7} \left( \sum_{j=1}^8 A_{1,j} x_j +\alpha_2 \frac{x_1}{x_2}(-x_{3}-x_{4}-x_{5}-x_{6}+x_{8}) \right),\\
\dot x_8&=x_{8} \left( \sum_{j=1}^8 A_{1,j} x_j +\alpha_2 \frac{x_1}{x_2}(-x_{3}-x_{4}-x_{5}-x_{6}-x_{7}) \right).
\end{align*}
This system admits six integrals: $\widehat{H}$ as given by \eqref{tilH8}, and
\begin{align*}
H&=\sum_{i=1}^8 x_i\\    
\widetilde{F}_1/H&=\frac{\left(x_{3}+x_{4}+x_{5}+x_{6}\right) x_{8}}{\left(x_{8}+x_{7}+x_{6}+x_{5}+x_{4}+x_{3}\right) x_{7}}\\
\widetilde{F}_2/H&=\frac{\left(x_{3}+x_{4}\right) x_{6} x_{8}}{\left(x_{8}+x_{7}+x_{6}+x_{5}+x_{4}+x_{3}\right) x_{5} x_{7}}\\
\widetilde{G}_1&=\frac{\left(x_{8}+x_{7}\right) x_{5} x_{3}}{\left(x_{8}+x_{7}+x_{6}+x_{5}+x_{4}+x_{3}\right) x_{6} x_{4}}\\
\widetilde{G}_2&=\frac{\left(x_{8}+x_{7}+x_{6}+x_{5}\right) x_{3}}{x_{4} \left(x_{8}+x_{7}+x_{6}+x_{5}+x_{4}+x_{3}\right)}.
\end{align*}
Similar to the system \eqref{4ds8} and its 7-dimensional variant, the system \eqref{even8} is Liouville integrable: the functions $\widehat{H},H$ commute with each other and with the tuples $(\widetilde{F}_1/H,\widetilde{F}_2/H)$ and $(\widetilde{G}_1,\widetilde{G}_2)$. It is also superintegrable of rank 2.
\end{exam}

By analogy to commutants in the context of Cartan subalgebras of simple Lie algebras, cf.  \cite{cam23,cam23c}, in the next two sections we shift our focus to Abelian subalgebras. In Section \ref{sec5}, we consider the generators of the subalgebra to be a commuting subset of those of the algebra of integrals. In Section \ref{sec6}, we consider single generator which is a multiparameter linear combination of generators of the algebra of integrals. In both cases we obtain Liouville (super)integrable deformations of \eqref{LV}. 

\section{Integrable deformations of Lotka-Volterra system \eqref{LV}, using an Abelian subalgebra of integrals}
\label{sec5}
We provide integrable deformations of the Lotka-Volterra system \eqref{LV} in arbitrary dimensions, using an Abelian subalgebra of integrals. In Section \ref{ods} we consider deformations of the odd-dimensional systems, and in Section
\ref{eds} we deform the even-dimensional cases.
\subsection{Deformation of the odd $n=2m-1$ dimensional system \eqref{LV}} \label{ods}
Let $m>1$ be an integer. We consider the space $\F=\fu(H,F_2,\ldots,F_{m-1},G_2,\ldots,G_{m-1},C)$, where the nontrivial brackets are given by \eqref{PBO}. For $l\in\{0,1,\ldots,m-4\}$, we define the subalgebra of integrals
\[
\G_l=\fu(H,F_{m-1-l},F_{m-l},\ldots,F_{m-1}).
\]
The constraints on the commutants relative to this subalgebra is given by the set of equations $\{F_i,L\}_\F=0$, $i=m-1-l,\ldots,m-1$,
which amounts to the following system of PDEs
\begin{equation} \label{spdes}
\sum_{j=2}^{m-1} \{F_i,G_j\}_\F L_{G_j} =0,\quad i=m-1-l,\ldots,m-1.   
\end{equation}
Solving the system \eqref{spdes}, using Maple, yields the rational functions
\begin{equation}\label{L}
L_k=\frac{F_{k +1} G_{m -k}-C H }{F_{k +1} G_{m-k-1}+CH},\quad k=1,\ldots,m-l-3.
\end{equation}
\begin{prop} \label{p14}
The commutant of $\G_l$ includes
$\fu(H,C,F_2,\ldots,F_{m-1},L_1,\ldots,L_{m-l-3})$.
\end{prop}
\begin{proof}
We prove that $\{F_i,L_k\}_\F=0$ for $m-1-l\leq i\leq m-1$ and $1\leq k \leq m-l-3$. The function $L_k$ depends on only two of the $G_j$, namely on $G_{m -k}$ and $G_{m -k-1}$. Hence, we find
\begin{align*}
\{F_i,L_k\}_\F&=\{F_i,G_{m -k}\}_\F\frac{\partial L_k}{\partial G_{m -k}} +
\{F_i,G_{m -k-1}\}_\F\frac{\partial L_k}{\partial G_{m -k-1}}\\
&=-\frac{\prod_{j=0}^{i-k-1} F_{k+j+1} G_{m -k-j} -CH}{\prod_{j=1}^{i-k-1} F_{k+j} G_{m -k-j} + CH}
\left(\frac{F_{k +1} }{\left(F_{k +1} G_{m-k-1} +C H\right)}\right)\\
&\ \ \ 
+\frac{\prod_{j=0}^{i-k-2} F_{k+j+2} G_{m-k-1-j} - CH}{\prod_{j=1}^{i-k-2}F_{k+j+1} G_{m -k-1-j} + CH}
\left(\frac{F_{k +1} \left(F_{k +1} G_{m-k} -C H\right)}{\left(F_{k +1} G_{m-k-1} +C H\right)^{2}}\right)\\
&=0,
\end{align*}
which can be seen by shifting the dummy variable $j\rightarrow j'=j+1$ in the second term, and including the factors with $j'=0$ in the numerator, and $j'=1$ in the denominator, in the products.  
\end{proof}
Using $\G_l$, we deform the Hamiltonian of the odd-dimensional Lotka-Volterra system \eqref{LV} as
\[
\widehat{H}=H+\alpha_1F_{m-1}+\alpha_2F_{m-2}+\cdots\alpha_{l+1}F_{m-1-l}.
\]
This yields a $2m-2$ dimensional Liouville integrable system. There are 2 Casimirs $H,C$, and $m-2$ integrals in involution $\widehat{H},F_2,\ldots,F_{m-2}$ (note that $F_{m-1}$ is not independent). We have $m-l-3$ additional integrals which make the system superintegrable, of ranks $l+1,l+2,\ldots,m-2$. The new algebra of integrals will be provided, and employed, in Section \ref{sec7}.

\begin{exam} \label{e15}
Let $n=11$, so that $m=6$. We consider the algebra of integrals $\F=\bbF(X)$, where $X=(H,F_2,F_3,F_4,F_{5},G_2,G_3,G_4,G_{5},C)$, and three different Abelian subalgebra of integralss $\G_0=\fu(H,F_5)$, $\G_1=\fu(H,F_4,F_5)$ and $\G_2=\fu(H,F_3,F_4,F_5)$.
\begin{itemize}
\item The subalgebra $\G_0=\fu(H,F_5)$ has commutants $(H,F_2,\ldots,F_5,L_1,L_2,L_3,C)$, with $L_k$ given by \eqref{L}. Taking as deformed Hamiltonian
$\widehat{H}=H+\alpha_1 F_5$, the system $\dot{X}=\{X,\widehat{H}\}_\F$ is given by
\begin{equation} \label{uio}
\begin{split}
\dot{H}&=\dot{F}_2=\dot{F}_3=\dot{F}_4=\dot{F}_5=\dot{C}=0\\
\dot{G}_2&=\alpha_1\left(F_{5} G_{2}-C H\right)\\
\dot{G}_3&=\alpha_1\frac{\left(F_{4} G_{3} -C H\right) \left(F_{5} G_{2}-C H \right)}{F_{4} G_{2}+C H }\\ 
\dot{G}_4&=\alpha_1\frac{\left( F_{3} G_{4}-C H\right) \left(F_{4} G_{3}-C H \right) \left(F_{5} G_{2}-C H \right)}{\left(F_{3} G_{3} +C H\right) \left(F_{4} G_{2}+C H \right)}\\
\dot{G}_5&=\alpha_1\frac{\left(F_{2} G_{5}-C H \right) \left(F_{3} G_{4}-C H \right) \left(F_{4} G_{3}-C H \right) \left(F_{5} G_{2}-C H \right)}{\left(F_{2} G_{4} +C H\right) \left(F_{3} G_{3} +C H\right) \left(F_{4} G_{2} +C H\right)}.
\end{split}
\end{equation}
This 10-dimensional system admits two sets of pairwise commuting integrals (including two Casimirs, $\{\widehat{H},F_2,F_3,F_4,H,C\}$ and $\{\widehat{H},L_1,L_2,L_3,H,C\}$.
We can show that the system does not admit other integrals than $\{F_2,F_3,F_4,F_5,L_1,L_2,L_3,H,C\}$ using the PDEs which constrain a commutant $K(X)$, namely $\{F_5,K(X)\}=0$, or, explicitly
\begin{align*}
\left(C H-F_5 G_2\right)K_{G_2}-&\frac{ \left(C H-F_5 G_2\right) \left(C H-F_4 G_3\right)}{F_4 G_2+C H}K_{G_3}+\frac{\left(C H-F_5 G_2\right) \left(C H-F_4 G_3\right) \left(C H-F_3 G_4\right)}{\left(F_4 G_2+C H\right) \left(F_3 G_3+C H\right)}K_{G_4}\\
-&\frac{\left(C H-F_5 G_2\right) \left(C H-F_4 G_3\right) \left(C H-F_3 G_4\right)}{\left(F_4 G_2+C H\right) \left(F_3 G_3+C H\right) \left(F_2 G_4+C H\right)}K_{G_5}=0,
\end{align*}
and the birational change of variables
\[   z_1=\frac{F_2 G_5-C H}{F_2 G_4+C H},z_2=\frac{F_3 G_4-C H}{F_3 G_3+C H},z_3=\frac{F_4 G_3-C H}{F_4 G_2+C H},z_4=G_2 , z_5=F_2, z_6=F_3, z_7=F_4, z_{8}=F_5. \]
In these variables the PDE reduces to
\[
(CH  -  z_4 z_8   )  K_{z_4}=0,
\]
which implies any commutant is a function of the known ones only. Corollary \ref{c19} states that the non-vanishing Poisson brackets between the integrals are as follows
\[
\{L_1,F_2\}_\F=L_1 F_2,\quad
\{L_2,F_3\}_\F=L_2 F_3,\quad
\{L_3,F_4\}_\F=L_3 F_4,\quad L_i\in\F.
\]  

Thus, the 10-dimensional system \eqref{uio} is superintegrable of rank 4 (Liouville-integrable), with commuting integrals
\[
(\widehat{H},F_2,F_3,F_4,H,C),\quad
(\widehat{H},L_1,L_2,L_3,H,C).
\]
It is superintegrable of rank 3 with tuples of integrals
\begin{align*}
&(\widehat{H},F_3,F_4,F_2,L_1,H,C),\quad
(\widehat{H},F_2,F_4,F_3,L_2,H,C),\quad
(\widehat{H},F_2,F_3,F_4,L_3,H,C),\\
&(\widehat{H},L_1,L_2,L_3,F_4,H,C),\quad
(\widehat{H},L_1,L_3,L_2,F_3,H,C),\quad
(\widehat{H},L_2,L_3,L_1,F_2,H,C).
\end{align*}
It is superintegrable of rank 2 with tuples of integrals
\begin{align*}
&(\widehat{H},F_4,F_2,F_3,L_1,L_2,H,C),\quad
(\widehat{H},F_3,F_2,F_3,L_1,L_3,H,C),\quad
(\widehat{H},F_2,F_3,F_4,L_2,L_3,H,C),\\
&(\widehat{H},L_1,L_2,L_3,F_3,F_4,H,C),\quad
(\widehat{H},L_2,L_1,L_3,F_2,F_4,H,C),\quad
(\widehat{H},L_3,L_1,L_2,F_2,F_3,H,C).
\end{align*}
And it is superintegrable of rank 1 with integrals $(\widehat{H},F_2,F_3,F_4,L_1,L_2,L_3,H,C)$. In terms of the initial variables the deformed Hamiltonian is
\[
\widehat{H}= \sum_{j=1}^{11}x_j + \alpha_1 \frac{x_{11}}{x_{10}} \sum_{j=1}^{9}x_j
\]
and the system is
\begin{equation} \label{xsy}
\dot{x}_i=x_i\cdot\begin{cases}
-\sum_{j=1}^{i-1}x_j+\sum_{j=i+1}^{11} x_j +\alpha_1  \frac{x_{11}}{x_{10}}\left(-\sum_{j=1}^{i-1}x_j+ \sum_{j=i+1}^{9} x_j\right) &i\leq 9,\\
-\sum_{j=1}^{i-1}x_j+\sum_{j=i+1}^{11} x_j &i=10,11.\\
\end{cases}
\end{equation}
It has the following $9=n-2$ integrals: $\widehat{H}$, $H$, $C$, $F_2$, $F_3$, $F_4$ and
\begin{align*}
K_1&=\frac{1}{L_1}=\frac{\left(x_{1}+x_{2}+x_{3}+x_{4}\right) x_{3}}{\left(x_{1}+x_{2}\right)x_{4 }},\\
K_2&=\frac{K_1}{L_2}=\frac{\left(x_{1}+x_{2}+x_{3}+x_{4}+x_{5}+x_{6}\right)x_{3} x_{5} }{\left(x_{1}+x_{2}\right)x_{4}x_{6}  },\\
K_3&=\frac{K_2}{L_3}=\frac{\left(x_{1}+x_{2}+x_{3}+x_{4}+x_{5}+x_{6}+x_{7}+x_{8}\right)x_{3} x_{5} x_{7} }{\left(x_{1}+x_{2}\right)x_{4} x_{6} x_{8} }.
\end{align*}
Both tuples $(\widehat{H},H,F_2,F_3,F_4,C)$ and $(\widehat{H},H,K_1,K_2,K_3,C)$ are pairwise commuting, the system is superintegrable of rank 5. Similar tuples as the ones mentioned above show the system is also superintegrable of ranks 4,3, and 2.
\item The subalgebra of integrals $\G_1=\fu(H,F_4,F_5)$  has commutants $(H,C,F_2,\ldots,F_5,L_1,L_2)$. The deformed Hamiltonian
$\widehat{H}=H+\alpha_1 F_5 + \alpha_2 F_4$ yields system \eqref{uio} with the following term added to its right hand side:
{\small \begin{equation} \label{at}
\left(0, \ldots , 0, \alpha_{2} \left(F_{4} G_{3}- C H \right), 
\frac{\alpha_{2} \left(C HF_{3} G_{4} -\right) \left(F_{4} G_{3}-C H \right)}{F_{3} G_{3} +C H}
, 
\frac{\alpha_{2} \left(F_{2} G_{5} -C H\right) \left(F_{3} G_{4}-C H\right) \left(F_{4} G_{3}-C H\right)}{\left(F_{2} G_{4}+C H\right) \left(F_{3} G_{3}+C H\right)}\right).
\end{equation}}
The system is  is superintegrable of rank 4 (Liouville-integrable), with commuting tuple
\[
(\widehat{H},F_2,F_3,F_4,H,C).
\]
It is superintegrable of rank 3 with tuples of integrals
\[
(\widehat{H},F_3,F_4,F_2,L_1,H,C),\quad
(\widehat{H},F_2,F_4,F_3,L_2,H,C).
\]
And it is superintegrable of rank 2 with integrals $(\widehat{H},F_4,F_2,F_3,L_1,L_2,H,C)$. Turning to the $x$-variables, we need to add the following terms to \eqref{xsy},
\begin{equation} \label{atx}
\alpha_2x_i\frac{x_9x_{11}}{x_8x_{10}}\cdot\begin{cases} -\sum_{j=1}^{i-1}x_j+\sum_{j=i+1}^{7} x_j & 1\leq i\leq 7,\\
0 & 8 \leq i \leq 11.\\
\end{cases}
\end{equation}
The resulting system admits the pairwise commuting tuple of integrals $(\widehat{H},H,F_2,F_3,F_4,C)$ as well as two additional integrals, $K_1,K_2$. It is superintegrable of ranks 5,4 and 3.
\item
Lastly, if take $\G_2=\fu(H,F_3,F_4,F_5)$ as the subalgebra of integrals, commutants are $(H,C,F_2,\ldots,F_5,L_1)$. The deformed Hamiltonian
$\widehat{H}=H+\alpha_1 F_5 + \alpha_2 F_4 + \alpha F_3$ yields system \eqref{uio} with \eqref{atx} added, as well as
\begin{equation} \label{att}
\left(0, \ldots , 0, \alpha_{3} \left(F_{3} G_{4} - CH\right), 
\frac{\alpha_{3} \left(F_{2} G_{5} - C H \right) \left(F_{3} G_{4}-C H \right)}{F_{2} G_{4} +C H}\right).
\end{equation}
The system is Liouville integrable, it admits the set of pairwise commuting integrals $\{\widehat{H}, H,F_2,F_3,F_4,C\}$ and has one additional integral, $L_1$.

The corresponding system in the $x$-variables is \eqref{xsy} with \eqref{atx}, and 
\begin{equation} \label{aatx}
\alpha_3x_i\frac{x_7x_9x_{11}}{x_6x_8x_{10}}\cdot \begin{cases} -\sum_{j=1}^{i-1}x_j+\sum_{j=i+1}^{5} x_j & 1\leq i\leq 5,\\
0 & 6 \leq i \leq 11.\\
\end{cases}
\end{equation}
The resulting system admits the pairwise commuting set of independent integrals $\{\widehat{H},H,F_2,F_3,F_4,C\}$ as well as an additional integral, $K_1$. It is superintegrable of ranks 5 and 4.
\end{itemize}
\end{exam}

\subsection{Deformation of the even $n=2m$ dimensional system \eqref{LV}} \label{eds}
For $1<m\in\mathbb{N}$, we consider the space $\F=\fu(X)$ with $X=(H,F_1,\ldots,F_{m-1},G_1,\ldots,G_{m-1})$, on which the non-trivial brackets are given by \eqref{PBE}. For $l\in\{0,1,\ldots,m-3\}$, we define the subalgebra of integrals
\[
\G_l=\fu(H,F_{m-1-l},\ldots,F_{m-1}).
\]
A function $J$ is a commutant of $\G_l$ if 
\begin{equation} \label{Jpde2}
\{F_i,J\}_\F = \sum_{j=1}^{m-1} \{F_i,G_j\}_\F J_{G_j} =0,\quad i=m-1-l,\ldots,m-1.
\end{equation}
\begin{prop}
The system of PDEs \eqref{Jpde2} is satisfied by
\begin{equation} \label{J}
J_k = F_{m-1-l} G_{1+k}-F_{m-1-k} G_{1+l}
\end{equation}
for each $k=l+1,\ldots,m-2$.
\end{prop}
\begin{proof}
We have, for $i=m-1-l,\ldots,m-1$ and $k=l+1,\ldots,m-2$,
\begin{align*}
\{F_i,J_k\}_\F&=-F_{m-1-k} \{F_i,G_{1+l}\}_\F + F_{m-1-l}  \{F_i,G_{1+k}\}_\F\\
&=F_{m-1-l}F_{m-k-1}G_{m-i}-F_{m-1-l}F_{m-k-1}G_{m-i}=0,
\end{align*}
as $\kappa(i,1+k)\geq0$ and either $\kappa(i,1+l)\geq 0$ or, when $i=m-1-l$, $\kappa(i,1+l)=-1$. 
\end{proof}
We use $\G_l$ to deform the Hamiltonian of the even-dimensional Lotka-Volterra system \eqref{LV},
\begin{equation}
\widehat{H}=H+\alpha_1F_{r-1}+\alpha_2F_{r-2}+\cdots\alpha_{l+1}F_{r-1-l}.
\end{equation}
This yields a $2r-1$ dimensional Liouville integrable and superintegrable system. We have $r-1$ integrals in involution, $\widehat{H},F_1,\ldots,F_{r-2}$ and a Casimir $H$, and $r-2-l$ additional integrals, $J_{l+1},\ldots,J_{r-2}$ (which are also pairwise commuting). The system is also superintegrable of rank $l+1$. 

\begin{exam}
Let $n=10$, so that $m=5$. We consider the 9-dimensional algebra of integrals $\F=\fu(X)$ with $X=(H,F_1,F_2,F_3,F_4,G_1,G_2,G_3,G_4)$, and three different Abelian subalgebra of integralss $\G_0=\fu(H,F_4)$, $\G_1=\fu(H,F_3,F_4)$ and $\G_2=\fu(H,F_2,F_3,F_4)$.
\begin{itemize}
\item The subalgebra $\G_0=\fu(H,F_4)$ has commutants $(H,F_1,F_2,F_3,F_4,J_1,J_2,J_3)$,
where
\[
J_{1} = F_{4} G_{2}-F_{3} G_{1},\quad
J_{2} = F_{4} G_{3}-F_{2} G_{1},\quad
J_{3} = F_{4} G_{4}-F_{1} G_{1}.
\]
There are no other commutants. This follows from writing the PDE
\[
\{F_4,K(X)\}_\F=-G_1 (F_4K_{G_1}+F_3K_{G_2}+F_2K_{G_3}+F_1K_{G_4}) =0
\]
in terms of the variables
\[
z_1=G_1, z_2=J_1, z_3=J_2, z_4=J_3, z_5=F_1, z_6=F_2, z_7=F_3, z_8=F_4.
\]
The inverse of this transformation is
\[
F_{1} = z_{5}, F_{2} = z_{6}, F_{3} = z_{7}, F_{4} = z_{8}, 
G_{1} = z_{1}, G_{2} = \frac{z_{1} z_{7}+z_{2}}{z_{8}}, G_{3} = 
\frac{z_{1} z_{6}+z_{3}}{z_{8}}, G_{4} = 
\frac{z_{1} z_{5}+z_{4}}{z_{8}},
\]
and the PDE reduces to $z_1 z_8 K_{z_1} =0$.
Taking as deformed Hamiltonian
$\widehat{H}=H+\alpha_1 F_4$, the system $\dot{X}=\{X,\widehat{H}\}_\F$ is given by
\begin{equation} \label{uie}
\begin{split}
\dot{H}&=0\\
\dot{F}_i&=0,\quad i=1,\ldots 4\\
\dot{G}_i&=\alpha_1 G_1 F_{5-i},\quad i=1,\ldots 4.
\end{split}
\end{equation}
The 9-dimensional system \eqref{uie} is superintegrable of rank 4 with tuples of commuting independent integrals, $(\widehat{H}, F_1,F_2,F_3, H)$ and $(\widehat{H}, J_1,J_2,J_3, H)$. The  system \eqref{uie} is also superintegrable of rank 1 as it admits 8 independent integrals $(\widehat{H},F_1,F_2,F_3,J_1,J_2,J_3,H)$.

In the original coordinates, we find
\begin{equation} \label{xsy2}
\dot{x}_i=x_{i}\cdot\begin{cases}
-\sum_{j=1}^{i-1}x_j+\sum_{j=i+1}^{10} x_j +\alpha_1  \frac{x_{10}}{x_{9}}\left(-\sum_{j=1}^{i-1}x_j+ \sum_{j=i+1}^{8} x_j\right) &i\leq 8,\\
-\sum_{j=1}^{i-1}x_j+\sum_{j=i+1}^{10} x_j  &i=9,10\\
\end{cases}
\end{equation}
which is superintegrable of ranks 5 and 2.
\item The subalgebra $\G_1=\fu(H,F_3,F_4)$ has commutants $(H,F_1,F_2,F_3,F_4,J_2,J_3)$, where
\[
J_{2} = F_{3} G_{3} - F_{2} G_{2},\quad
J_{3} = F_{3} G_{4} - F_{1} G_{2}.
\] Taking as deformed Hamiltonian
$\widehat{H}=H+\alpha_1 F_4 + \alpha_2 F_3$, the system $\dot{X}=\{X,\widehat{H}\}_\F$ is given by
\begin{equation} \label{uie2}
\begin{split}
\dot{H}&=0\\
\dot{F}_i&=0,\quad i=1,\ldots 4\\
\dot{G}_1&=G_{1} \left(\alpha_{1} F_{4}+\alpha_{2} F_{3}\right)\\
\dot{G}_i&=F_{5-i} \left(\alpha_{1}G_{1} +\alpha_{2}G_{2} \right),\quad i=2,3,4.
\end{split}
\end{equation}
System \eqref{uie2} is superintegrable of rank 4 with commuting independent integrals
$(\widehat{H}, F_1,F_2,F_3,H)$. It is superintegrable of rank 2 with integrals
$(\widehat{H}, F_3,F_2,F_1,J_2,J_3,H)$.  

The corresponding $x$-system is \eqref{xsy2}, with the following terms added to the right-hand-side:
\begin{equation} \label{atx2}
\alpha_2x_i\frac{x_8x_{10}}{x_7x_{9}}\begin{cases} -\sum_{j=1}^{i-1}x_j+\sum_{j=i+1}^{6} x_j & 1\leq i\leq 6,\\
0 & 7 \leq i \leq 10.\\
\end{cases}
\end{equation}
The system \eqref{atx2} is superintegrable of ranks 5 and 3.
\item The subalgebra $\G_2=\fu(H,F_2,F_3,F_4)$ has commutant $\{H,F_1,F_2,F_3,F_4,J_3\}$, where
\[
J_{3} = F_{2} G_{4}-F_{1} G_{3}.
\]
Taking as deformed Hamiltonian
$\widehat{H}=H+\alpha_1 F_4 + \alpha_2 F_3 + \alpha_3 F_2 $, the system $\dot{X}=\{X,\widehat{H}\}_\F$ is given by
\begin{equation} \label{uie3}
\begin{split}
\dot{H}&=0\\
\dot{F}_i&=0,\quad i=1,\ldots 4\\
\dot{G}_1&=G_{1} \left(\alpha_{1} F_{4}+\alpha_{2} F_{3}+\alpha_{3} F_{2}\right)\\
\dot{G_2}&=\alpha_{1} F_{3} G_{1}+\alpha_{2} F_{3} G_{2} + \alpha_{3} F_{2} G_{2}\\
\dot{G}_i&=F_{5-i} \left(\alpha_{1}G_{1} +\alpha_{2}G_{2}+\alpha_{3}G_{3} \right),\quad i=3,4.
\end{split}
\end{equation}
System \eqref{uie3} is superintegrable of rank 4 with commuting independent integrals
$(\widehat{H}, F_1,F_2,F_3,H)$. It is superintegrable of rank 3 with integrals
$(\widehat{H}, F_3,F_2,F_1,J_3,H)$. The $x$-system, \eqref{xsy2} with \eqref{atx2} as well as the following terms added to the right-hand-side:
\begin{equation}
\alpha_3x_i\frac{x_6x_8x_{10}}{x_5x_7x_{9}}\cdot\begin{cases} -\sum_{j=1}^{i-1}x_j+\sum_{j=i+1}^{4} x_j & 1\leq i\leq 4,\\
0 & 5 \leq i \leq 10,\\
\end{cases}
\end{equation}
is superintegrable of ranks 5 and 4.
\end{itemize}
\end{exam}

\section{Integrable deformations of Lotka-Volterra system \eqref{LV}, using a multi-parameter one-dimensional subalgebra of integrals} \label{sec6}

In Sections \ref{sec4} and \ref{sec5}, the commutant was taken relative to subalgebra of integralss that were non-Abelian and Abelian. In this section, we will consider another case where the subalgebra of integrals is a linear combination of the generators of the algebra of integrals. For even $n$, the general linear combination is
\begin{equation}
  L=\sum_{i=1}^{m-1} \alpha_i F_i +  \sum_{i=1}^{m-1} \beta_i G_i, 
\end{equation}  
and for odd $n$ it is
\begin{equation}
  L= \sum_{i=2}^{m-1} \alpha_i F_i +  \sum_{i=2}^{m-1} \beta_i G_i. 
\end{equation}  
where we have excluded the terms with $F_1=G_1=C$. As we will see, this choice will provide us with integrals beyond polynomials and rational functions.

\subsection{Odd case (linear combination of commuting functions)}
We consider the space $\F=\fu(X)$, where $X=(H,F_2,\ldots, F_{m-1}, G_2, \ldots,G_{m-1},C)$ and the nontrivial brackets are given by \eqref{PBO}. As one-dimensional subalgebra of integrals we take $\G=\fu(F_2 + a F_3)$, whose commutants $K(X)$ satisfies the PDE
\begin{align*}
  0&=\{K,  F_2 + a F_3\}_\F\\
  &=  - a  \left( C H-F_3 G_{m-2}\right) K_{G_{m-2}}
  +\left(a \frac{CH-F_3 G_{m-2}}{F_2 G_{m-2}+ C H} -1\right)\left(C H-F_2 G_{m-1}\right)K_{G_{m-1}}.
\end{align*}  
which admits the following functionally independent solutions:
$(H,C,F_2,\ldots,F_{m-1},G_2,\ldots,G_{m-3},K)$, where
\begin{equation}
 K=\frac{\left(F_2 G_{m-1}-CH\right) \left(F_3 G_{m-2}-CH\right)^{-F_2/(a F_3)}}{F_2 \left(F_2 G_{m-2}+ C H\right)} 
\end{equation} 
is a transcendental function. Note that the functions $G_{m-2}$ and $G_{m-1}$ are no longer integrals of any deformed Hamiltonian
\begin{equation}
 \widehat{H}= \widehat{H}(H,F_2 + a F_3)  
\end{equation} 
The non-vanishing Poisson brackets between the new integral $K$ and the integrals of type $F,G$ are
\begin{equation}
  \{K,F_2\}_\F= K F_2  ,\quad  \{K,F_s\}_\F= -\frac{F_2K}{a} \prod_{i=3}^{s-1}  \frac{F_{i+1} G_{m-i}-CH}{F_i G_{m-i}+ C H} , \quad s \geq 3.
\end{equation}  

\subsection{Even case (linear combination of commuting functions)}
We consider the space $\F=\fu(X)$, where $X=(H,F_1,\ldots,F_{m-1},G_1, \ldots, G_{m-1})$ and the nontrivial brackets are given by \eqref{PBE}, with subalgebra of integrals $\fu(F_1 + a  F_2)$. We determine the functions $J(X)$ that satisfy
\[\{J, F_1 + a  F_2\}_\F= \sum_{i=1}^{m-2} G_i  (F_1 + a F_2 ) J_{G_i}+  (F_1 G_{m-1} + a F_2 G_{m-2} )J_{G_{m-1}}=0.
\]
Functionally independent solutions are given by 
$(H,F_1,\ldots,F_{m-1},\widetilde{G}_2,\ldots,\widetilde{G}_{m-2},J)$, where
\[
\widetilde{G}_i=\dfrac{G_i}{G_1},\quad J=\left(G_{m-1}-\frac{F_1 G_{m-2}}{F_2}\right)
G_1^{-F_1/(a F_2+F_1)}.
\]
We observe the $G_i$ are no longer integrals of the deformed Hamiltonian, but we obtain integrals $\widetilde{G}_2,\ldots,\widetilde{G}_{m-2}$ which allows the deformation to be superintegrable. The Poisson brackets are $\{\widetilde{G}_i,J\}_\F=0$, and 
\[
\{F_i,\widetilde{G}_j\}_\F=\begin{cases}
      F_{i}\widetilde{G}_{j}-F_{m-j}\widetilde{G}_{m-i}, &i+j>m,\\
      0, &i+j\leq m,
  \end{cases}\quad
  \{F_i, J\}_\F= \begin{cases}
      -a F_{1} F_{2}J/(F_{1}+a F_{2}), &i=1,\\
      F_{1} F_{i}J/(F_{1}+a F_{2}), &i>1.
  \end{cases}
\]

\subsection{Odd Case (linear combination of noncommuting functions)}
We consider the space $\F=\fu(X)$, where $X=(H,F_2,\ldots, F_{m-1}, G_2, \ldots,G_{m-1},C)$ and the nontrivial brackets are given by \eqref{PBO}. As one-dimensional subalgebra of integrals we take $\G=\fu(F_2 + a F_3 + b G_{m-2} )$ and look for commutants. The corresponding PDE admits the following pairwise commuting sets of solutions:
\[
\widetilde{F}_i= \begin{cases}
    F_2  &  i=2 \\
    a F_3 + b G_{m-2}  & i=3 \\[1mm]
    \dfrac{F_i G_{m+1-i}-CH}{F_{i-1} G_{m+1-i}+CH} & i=4,\ldots, m-3,
\end{cases}
\quad
\widetilde{G}_i= \begin{cases}
    G_i  &  i=2,\ldots,m-3 \\[1mm]
    \dfrac{C H -F_{2} G_{m -1}}{bF_{2} \left(C H +F_{2} G_{m -2}\right)} \mathrm{e}^Q & i=m-2, 
\end{cases}
\]
where
\[
Q=-\frac{2 F_{2}}{P} \arctan\left(\frac{a F_{3}-b G_{m -2}}{P}\right),\quad P=\sqrt{4 a b C H  -\left(a F_{3}+b G_{m-2}\right)^{2}}.
\]
The non-vanishing Poisson brackets are:
\[
\{\widetilde{F}_i,\widetilde{G}_j\}_\F= \begin{cases}
-\widetilde{F}_i \widetilde{G}_j & i=2, j=m-2 \text{ or } i>3, j=m+1-i, \\
\widetilde{F}_{i-1}\widetilde{G}_{j} & i=3, j=m-2.
\end{cases}
\]

\subsection{Even case (linear combination of noncommuting functions)}
We consider the space $\F=\fu(X)$, where $X=(H,F_1,\ldots,F_{m-1},G_1, \ldots, G_{m-1})$ and the nontrivial brackets are given by \eqref{PBE}, with subalgebra of integrals $\fu(F_1 + a  F_2 + b G_1)$. Pairwise commuting sets of commutants are given by
\[
\widetilde{F}_i= \begin{cases}
F_1 + a  F_2 + b G_1 & i=1, \\
    F_i/F_1  &  i=2,..,m-1,
\end{cases}
\quad
\widetilde{G}_i= \begin{cases}
\left(G_{m-1}-\frac{F_1}{F_2}G_{m-2}\right)(-bG_1)^{-F_1/(aF_2 + F_1)} 
& i=1, \\
    G_i/G_1 &  i=2,\ldots,m-2.
\end{cases}
\] 
The non-vanishing Poisson brackets between these functions are given by
\[
\{\widetilde{F}_i,\widetilde{G}_j\}_\F = \begin{cases}
\widetilde{F}_i\widetilde{G}_j, &i=2,\ldots,m-1,\ j=1,\\
(F_iG_j-F_{m-j}G_{m-i})/F_1/G_1 &i>m-j,\ j>1.\\
\end{cases}
\]

\section{Iterative application of the method} \label{sec7}
We already mentioned the possibility of applying iteratively the construction, i.e. once a new algebra of integrals and related deformed Hamiltonian are obtained, they can then be used to further deform the system while retaining superintegrability. We provide the algebra of integrals for the commutants found in Section \ref{ods}, and illustrate how to use them to deform the deformed systems.

\smallskip
We consider the space $\F=\fu(H,F_2,\ldots,F_{m-1},G_2,\ldots,G_{m-1},C)$, where the nontrivial brackets are given by \eqref{PBO}. 
\begin{prop}
The non-vanishing Poisson brackets between the commutants $\{F_2,\ldots,F_{m-1},L_1,\ldots,L_{m-l-3}\}$, see Prop. \ref{p14} and eq. \eqref{L}, are
\[
\{F_{j+1},L_j\}_\F=-F_{j+1} L_{j},\quad j=1,\ldots,m-3.
\]
\end{prop}
\begin{proof}
We first consider
\begin{align*}
\{F_i,L_j\}_\F &= \left\{F_i,\frac{F_{j +1} G_{m -j}-C H }{F_{j +1} G_{m-j-1}+CH}\right\}_\F\\
&= \frac{ F_{j+1} }{ ( F_{j+1} G_{m-j-1} + CH ) }\{F_i,G_{m-j}\}_\F + \frac{( F_{j+1} G_{m-j} - CH  )F_{j+1}}{( F_{j+1} G_{m-j-1} + CH )^2} \{F_i , G_{m-j-1} \}_\F.
\end{align*}
Based on the value of \eqref{kij}, there are three cases: in both Poisson brackets $\kappa <0$, both $\kappa \geq 0$ or one of the Poisson brackets is $\kappa <0$ and the other is $\kappa \geq 0$. The first two cases give directly $0$. The third case occurs only when $i=j+1$, and then
\[
\{F_{j+1},L_j \}_\F = \frac{ F_{j+1} }{F_{j+1} G_{m-j-1} + CH  }\{F_{j+1} , G_{m-j} \}_\F  = F_{j+1} \frac{  CH - F_{j+1} G_{m-j} }{F_{j+1} G_{m-j-1} + CH} = - F_{j+1} L_j.  \]
Next, we consider
\begin{align*}
 \{L_i,L_j\}_\F&= \left\{ \frac{F_{i +1} G_{m -i}-C H }{F_{i +1} G_{m-i-1}+CH},  \frac{F_{j +1} G_{m -j}-C H }{F_{j +1} G_{m-j-1}+CH}\right\}_\F\\
 &= \frac{\partial L_i}{\partial F_{i+1}} \frac{\partial L_j}{\partial G_{r-j}} \{F_{i+1},G_{r-j}\}_\F  +   \frac{\partial L_i}{\partial F_{i+1}} \frac{\partial L_j}{\partial G_{r-j-1}} \{F_{i+1},G_{r-j-1}\}_\F\\
 &\ \ \ +   \frac{\partial L_i}{\partial G_{r-i}} \frac{\partial L_j}{\partial F_{j+1}} \{G_{r-i},F_{j+1}\}_\F +  \frac{\partial L_i}{\partial G_{r-i-1}} \frac{\partial L_j}{\partial F_{j+1}} \{G_{r-i-1},F_{j+1}\}_\F.   
\end{align*}
The above brackets have $\kappa=i-j$, $\kappa=i-j-1$, $\kappa=j-i$ and $\kappa=j-i-1$. Assume, without loss of generality, that $i<j$. Then the first two terms vanish, and we are left with
\[
   \{L_i,L_j\}_\F=\frac{\partial L_j}{\partial F_{j+1}}\left(\frac{\partial L_i}{\partial G_{r-i}}  \{G_{r-i},F_{j+1}\} +  \frac{\partial L_i}{\partial G_{r-i-1}}  \{G_{r-i-1},F_{j+1}\}\right)=0,
\]
due to
\[
\frac{\partial L_i}{\partial G_{r-i-1}} = \frac{ F_{i+1} G_{r-i} -CH}{(F_{i+1}G_{r-i-1}+CH)^2} F_{i+1},\quad \frac{\partial L_i}{\partial G_{r-i}}= \frac{ F_{i+1}}{F_{i+1} G_{r-i-1}+CH}
\]
and the definition of the bracket \eqref{PBO}.
\end{proof}

\begin{cor} \label{c19}
  The Poisson algebra of integrals of the integrals
\[
(\widehat{H}=H+\sum_{i=0}^{l} \alpha_{i+1}F_{5-i},F_2,F_3,F_4,L_1,\ldots,L_{3-l},H,C)
\]
in Example \ref{e15}, with $l=0$, has non-vanishing brackets
\begin{equation}\label{LF}
\{L_1,F_2\}_\F=L_1 F_2,\quad
\{L_2,F_3\}_\F=L_2 F_3,\quad
\{L_3,F_4\}_\F=L_3 F_4,
\end{equation}
where the $L_i$ are functions in $\F$.
\end{cor}

In the next example, we illustrate how to deform a deformed system, considering the case $l=0$ in Example \ref{e15}.

\begin{exam}
We consider the space ${\cal L}=\fu(L)$, where $L=(\widehat{H},F_2,F_3,F_4,L_1,L_2,L_{3},H,C)$, equipped with a Poisson bracket defined by eq. \eqref{LF}, i.e.
\[
\{L_1,F_2\}_{\cal L}=L_1 F_2,\quad
\{L_2,F_3\}_{\cal L}=L_2 F_3,\quad
\{L_3,F_4\}_{\cal L}=L_3 F_4.
\]
The subalgebra of integrals generated by $g=aF_2+bL_2$ has independent commutants
\[
\widehat{\widehat{H}}=\widehat{H}+g,F_2,F_4,L_3,Q=L_1F_3^{aF_2/(bL_2)},\widehat{H},H,
\]
which are solutions to the PDE $\{g,Q\}_{\cal Z}=0$. The tuplets
\[
(\widehat{\widehat{H}},F_2,F_4,\widehat{H},H,C),\quad (\widehat{\widehat{H}},F_2,L_3,\widehat{H},H,C),\quad (\widehat{\widehat{H}},F_4,Q,\widehat{H},H,C),\quad (\widehat{\widehat{H}},L_3,Q,\widehat{H},H,C)
\]
are superintegrable of rank 3 (Liouville integrable). The tuplets
\[
(\widehat{\widehat{H}},F_2,F_4,L_3,\widehat{H},H,C),\quad (\widehat{\widehat{H}},F_4,F_2,Q,\widehat{H},H,C),\quad (\widehat{\widehat{H}},L_3,F_2,Q,\widehat{H},H,C),\quad (\widehat{\widehat{H}},Q,F_4,L_3,\widehat{H},H,C)
\]
are superintegrable of rank 2. And the tuplet
\[
(\widehat{\widehat{H}},F_2,F_4,L_3,Q,\widehat{H},H,C)
\]
is superintegrable of rank 1. In terms of the original $x$-variables, the system is \eqref{xsy} with added terms to the right hand side
\begin{equation} \label{ata}
a x_i\frac{x_{5} x_{7} x_{9} x_{11}}{x_{4} x_{6} x_{8} x_{10}}\cdot\begin{cases} -\sum_{j=1}^{i-1}x_j+\sum_{j=i+1}^{3} x_j & 1\leq i\leq 3,\\
0 & 4 \leq i \leq 11.\\
\end{cases}
\end{equation}
and
\begin{equation} \label{atb}
\frac{bx_{6} }{x_{5} \left(x_{1}+x_{2}+x_{3}+x_{4}+x_{5}+x_{6}\right)^{2}}\cdot\begin{cases}
-x_{1}^{2} \left(x_{5}+x_{6}\right)&i=1,\\ 
-x_{2} \left(x_{5}+x_{6}\right) \left(2 x_{1}+x_{2}\right)&i=2,\\
-x_{3} \left(x_{5}+x_{6}\right) \left(2 x_{1}+2 x_{2}+x_{3}\right)&i=3,\\ 
-x_{4} \left(x_{5}+x_{6}\right) \left(2 x_{1}+2 x_{2}+2 x_{3}+x_{4}\right)
&i=4,\\ 
x_{5} \left(x_{1}+x_{2}+x_{3}+x_{4}\right) \left(x_{1}+x_{2}+x_{3}+x_{4}-x_{6}\right)
&i=5,\\
x_{6} \left(x_{1}+x_{2}+x_{3}+x_{4}\right) \left(x_{1}+x_{2}+x_{3}+x_{4}+x_{5}\right) &i=6,\\
0 & 6<i.
\end{cases}
\end{equation}
It is superintegrable of ranks 3,4,5.
\end{exam}

\section{Conclusion} \label{sec8}
We have provided different deformations of superintegrable systems based on the notion of commutant or partial Casimirs. This concept has been applied in recent years in the context of Lie and Poisson-Lie algebras with the aim of providing new insight into algebraic superintegrable systems, their algebra of integrals, but also in the decomposition of enveloping algebras of Lie algebras, the missing label problem, and dynamical symmetries. This paper has greatly extended the construction by allowing the underlying Poisson algebra to involve polynomial and even rational bracket relations. 

We considered subalgebras of different types, Abelian, non-Abelian and one-dimensional subalgebras with parameters, for the construction of the commutant. The new integrals are  obtained as solutions to partial differential equations, and they can be polynomial, rational, or even transcendental functions. The Hamiltonians are by construction deformations of the initial Lotka-Volterra systems and they open the way for applications in a different context. This paper points out how, from a given superintegrable system, large families of deformed superintegrable systems can be obtained in arbitrary dimensions. In all examples provided, superintegrability is preserved under subalgebra of integrals deformation.

\section*{Acknowledgement}
The authors would like to thank Professor Pol Vanhaecke for enlightening discussions. Ian Marquette was supported by the Australian Research Council Future Fellowship FT180100099. 

\appendix

\section{Birational transformation} \label{secA}
We conjecture that the transformation from the $n$ variables $x_i$ to the $n-1$ integrals and distinguished variable $x_n$ is a birational map. We treat the odd and the even dimensional case separately.

\subsection{The odd dimensional case}
Taking $n=2m-1$, we define
\[
z=[F_1,F_2,\ldots,F_m,G_2,G_3,\ldots,G_{m-1},x_n].
\]
The aim is to find a formula for $x_i(z)$, $i=1,2,\ldots,n$.
We introduce intermediate coordinates,
\begin{equation} \label{ic}
y_i=\begin{cases}
\frac{x_{n-i}}{x_{n-i+1}} &i=1,\ldots,n-1, \\
x_i &i=n.
\end{cases}
\end{equation}
One can show that in these coordinates the functions $z$ are polynomial. In terms of the finite continued products
\begin{equation} \label{fcp}
Y_{a,b}=(((y_b+1)y_{b-1}+1)y_{b-2}+\cdots+1)y_a+1,
\end{equation}
we have
\begin{equation} \label{y(z)}
z_i=\begin{cases}
 Y_{n-2(i-1),n-1} y_n \prod_{j=1}^{m-i}y_{2j}& i\leq m,\\
Y_{1,2(i-m)} y_n \prod_{j=i-m+1}^{m-1}y_{2j} & m<i<n,\\
y_n & i=n.
\end{cases}
\end{equation}
We define two rational functions
\[
q_{-1}=\frac{z_m-z_n}{z_m(z_1-z_m)},\quad q_0=\frac{z_1(z_{m-1}+z_{m})}{z_m(z_1+z_{m-1})},
\]
and, recursively, $n-2$ polynomials, $q_1=z_{m-1}+z_{n}$ and, for $1<i<n-1$,
\[
q_i=\begin{cases}
-z_{m+k}q_{i-1}+z_1(z_{m}-z_n)\prod_{j=2}^k (z_1z_m-z_{m+1-j}z_{m-1+j})& i=2k, \\
z_{m-1-k}q_{i-1}+(-1)^kz_1z_n(z_{m-1}+z_m)\prod_{j=2}^k (z_1z_m+z_{m-j}z_{m-1+j})& i=2k+1.
\end{cases}
\]
In terms of these functions, we define
\[
y_i=\begin{cases}
-\frac{q_{i-2}}{q_i}\left((-1)^iz_1z_m
+z_{\lfloor (n-i+2)/2 \rfloor}
z_{\lfloor (n+i)/2 \rfloor}\right) &0<i<n-1,\\
(-1)^m\frac{q_{n-3}}{z_n(z_{m-1}+z_m)\prod_{j=2}^{m-2}
(z_1z_m+z_{m-j}z_{m+j-1})}&i=n-1,\\
z_n&i=n,
\end{cases}
\]
and claim they satisfy \eqref{y(z)}. And then we find
\begin{equation} \label{xj}
x_i=y_n\prod_{j=1}^{n-i}y_j, 
\end{equation}
which is what we were after. The statement has been verified using Maple \cite{Map} for $m\leq 7$.
\subsection{The even-dimensional case}
Taking $n=2m$, we define
\[
z=[F_1,F_2,\ldots,F_m,G_2,G_3,\ldots,G_{m-1},x_n].
\]
In terms of \eqref{ic} and \eqref{fcp} we have:
\begin{equation} \label{y(z)e}
z_i=\begin{cases}
Y_{n-2i+1,n-1} y_n\prod_{j=1}^{m-i}y_{2j} & i\leq m,\\
Y_{1,2(i-m)-1} y_n\prod_{j=i-m}^{m-1}y_{2j+1} & m<i<n,\\
y_n & i=n.
\end{cases}
\end{equation}
We claim that \eqref{y(z)e} is satisfied for
{\scriptsize
\begin{align*}
y_1&= \frac{z_{m}-z_{n}}{z_{m -1}+z_{n}}\\
y_2&=\begin{cases}
\frac{z_{1} \left(z_{2}+z_{1}\right)}{z_{1} z_{3}+z_{1} z_{4}+z_{2} z_{4}+z_{3} z_{4}} & m=2\\
-\frac{\left(z_{m -1}+z_{m}\right) \left(z_{m -2} z_{m +1}-z_{m -1} z_{m +2}\right)}{z_{n} z_{m +1} \left(z_{m}-z_{m -2}\right)+z_{n} z_{m +2} \left(z_{m}+z_{m -1}\right)+z_{m} z_{m +1} \left(z_{m -2}+z_{m -1}\right)} & m>2
\end{cases}\\
y_{2 i +1} &= 
\frac{z_{n} z_{m +i} \left(z_{m}+z_{m -1}\right)+z_{m +1} z_{m -i} \left(z_{m}-z_{n}\right)}{z_{n} z_{m +i +1} \left(z_{m}+z_{m -1}\right)+z_{m +1} z_{m -i -1} \left(z_{m}-z_{n}\right)},\quad 1\leq i \leq m-2 \\
y_{2 i} &= 
\frac{\left(z_{m +i} z_{m -i -1}-z_{m -i} z_{m +i +1}\right) \left(z_{n} \left(z_{m -1}+z_{m}\right) \left(z_{m +i}+z_{m +i -1}\right)+z_{m +1} \left(z_{m}-z_{n}\right) \left(z_{m -i}+z_{m -i +1}\right)\right)}{\left(z_{m -i} z_{m +i -1}-z_{m +i} z_{m -i +1}\right)\left(z_{n} \left(z_{m -1}+z_{m}\right) \left(z_{m +i}+z_{m +i +1}\right)+z_{m +1} \left(z_{m}-z_{n}\right) \left(z_{m -i}+z_{m -i -1}\right)\right)},\quad 2\leq i \leq m-2 \\
y_{n -2} &= 
-\frac{z_{1} z_{m} \left(z_{m +1} \left(z_{m}-z_{n}\right) \left(z_{1}+z_{2}\right)+z_{n} \left(z_{n -1}+z_{n -2}\right) \left(z_{m -1}+z_{m}\right)\right)}{\left(z_{1} z_{n -2}-z_{2} z_{n -1}\right)\left(z_{1} z_{m +1} \left(z_{m}-z_{n}\right)+z_{n} \left(z_{m}+z_{n -1}\right) \left(z_{m}+z_{m -1}\right)\right) },\quad m\neq 2\\
y_{n -1} &= 
\frac{z_{1} z_{m +1} \left(z_{m}-z_{n}\right)+z_{n} z_{n -1} \left(z_{m}+z_{m -1}\right)}{z_{m} z_{n} \left(z_{m -1}+z_{m}\right)}\\
y_n&=z_n,
\end{align*}}
and then \eqref{xj} holds. The statement has been verified using Maple \cite{Map} for $m\leq 8$.

\section{Second proof of Proposition \ref{PAs}} \label{secB}
\begin{lem} \label{AI}
The involution \eqref{istar} negates the bracket,
\[
\is\{A,B\}=-\{\is A,\is B\}.
\]
\end{lem}
\begin{proof}
It suffices to observe that for all $i,j$ we have $\is \{x_i,x_j\}=-\{\is x_i,\is x_j\}$.
\end{proof}
The involution \eqref{istar} induces an involutions on the space $\F_n$, given by \eqref{Fe} or \eqref{Fo}, as
\[
\is(F_i)=G_i,\quad \is(G_i)=F_i,\quad \is(H)=H,\quad \is(C)=C,
\]
which respects Lemma \ref{AI}. To prove Proposition \ref{PAs}(i), we show that the Jacobi-identity is satisfied, i.e. that
\begin{equation} \label{JI}
\{A,\{B,C\}\} + \{B,\{C,A\}\} + \{C,\{A,B\}\} =0,
\end{equation}
for all functions $A,B,C$. It suffices to show that \eqref{JI} holds for all choices of $A,B,C\in \{F_1,\ldots,F_{r-1},G_1,\ldots,G_{r-1},H\}$. If either one (or more) of $A,B,C$ equals $H$, or $A=F_i,B=F_j,C=F_k$, or $A=G_i,B=G_j,C=G_k$ for some $i,j,k$, then \eqref{JI} is trivially satisfied as each of the terms vanishes. Rest us to consider the following cases:
\begin{equation}\label{2FG}
\{F_i,\{F_j,G_k\}\} + \{G_k,\{F_i,F_j\}\} + \{F_j,\{G_k,F_i\}\} =0,
\end{equation}
and
\begin{equation}\label{2GF}
\{G_i,\{G_j,F_k\}\} + \{F_k,\{G_i,G_j\}\} + \{G_j,\{F_k,G_i\}\} =0,
\end{equation}
where the middle terms vanish. We claim that
\begin{equation} \label{iop}
\{F_i,\{F_j,G_k\}\}=\begin{cases}
    -F_i\{F_j,G_k\} & i\leq j \\
    -F_j\{F_i,G_k\} & j\leq i,
\end{cases}
\end{equation}
from which \eqref{2FG} follows directly. Equation \eqref{iop} can be proven as follows.
\begin{itemize}
\item Let $i\leq j$. If $\kappa(j,k)<0$ then $\kappa(i,k)<0$ and we have
\[
\{F_i,\{F_j,G_k\}\}=\{F_i,-F_jG_k\}=-F_j\{F_i,G_k\}=-F_jF_iG_k=-F_i\{F_j,G_k\}.
\]
If $\kappa(j,k)\geq0$ then, since we know $\kappa(i,r-j)<0$,
\[
\{F_i,\{F_j,G_k\}\}=\{F_i,-F_{r-k}G_{r-j}\}=-F_{r-k}\{F_i,G_{r-j}\}=-F_{r-k}F_iG_{r-j}=-F_i\{F_j,G_k\}.
\]
\item Next, let $j\leq i$. If $\kappa(j,k)<0$ then
\[
\{F_i,\{F_j,G_k\}\}=\{F_i,-F_jG_k\}=-F_j\{F_i,G_k\}.
\]
If $\kappa(j,k)\geq0$ (and hence $\kappa(i,k)\geq0$) then, for $j=i$ (so that $\kappa(i,r-j)<0$)
\[
\{F_i,\{F_j,G_k\}\}=\{F_i,-F_{r-k}G_{r-j}\}=
-F_{r-k}\{F_i,G_{r-j}\}=F_{r-k}F_iG_{r-j}
=-F_i\{F_j,G_k\}=-F_j\{F_i,G_k\},
\]
and for $j<i$ (so that $\kappa(i,r-j)\geq0$
\[
\{F_i,\{F_j,G_k\}\}=\{F_i,-F_{r-k}G_{r-j}\}=
-F_{r-k}\{F_i,G_{r-j}\}=F_{r-k}F_jG_{r-i}
=-F_j\{F_i,G_k\}.
\]
\end{itemize}
Finally, the identity \eqref{2GF} follows by application of $i^*$ to \eqref{2FG} and using Lemma \ref{AI}.

\smallskip
Next we prove Proposition \ref{PAs}(ii). The strategy is similar: it suffices to show \eqref{2FG}, i.e.
\begin{equation}\label{2FG2}
\{F_i,\{F_j,G_k\}\} + \{F_j,\{G_k,F_i\}\} =0,
\end{equation}
from which \eqref{2GF} follows by Lemma \ref{AI}. We claim that 
\begin{equation} \label{FFG}
\{F_i,\{F_j,G_k\}\}=\begin{cases}
0 &\text{if } \kappa(i,k) < 0 \text{ or } \kappa(j,k) < 0 \\
-F_i\{F_j,G_k\} &\text{if } \kappa(i,k) \geq 0  \text{ and } i\leq j \\
-F_j\{F_i,G_k\} &\text{if } \kappa(j,k) \geq 0  \text{ and } i\geq j,
\end{cases}
\end{equation}
which is sufficient for (\ref{2FG2}). To prove \eqref{FFG} we distinguish four cases.
\begin{itemize}
\item
If $\kappa(j,k) < 0$ then $\{F_j,G_k\}=0$.
\item
Let $\kappa(i,k) < 0$. The function $Q(F,G)=\{F_j,G_k\}$ only depends on $G_l$ with $l\leq k$. As $\kappa(i,l) < 0$ then $\{F_i,Q\}=0$.
\item Let $\kappa(i,k) \geq 0$ and $i\leq j$. We first consider $\kappa(i,k) = 0$, that is, $k=r+1-i$. We have
\[
\{F_j,G_k\}=
(-1)^{j-i} (CH -F_{i} G_{r+1-i} )Q,\qquad Q=\prod_{l=1}^{j-i} \dfrac{CH -F_{i + l} G_{r+1-i-l}}{CH +F_{i + l-1} G_{r+1-i-l}}.
\]
As the product $Q$ depends on $G_{r+1-i-l}$ with $l>0$ it commutes with $F_i$. Hence
\begin{align*}
\{F_i,\{F_j,G_k\}\}&=-F_{i}\{F_i,G_{r+1-i}\}(-1)^{j-i}\prod_{l=1}^{j-i} \dfrac{CH -F_{i + l} G_{r+1-i-l}}{CH +F_{i + l-1} G_{r+1-i-l}}\\
&=-F_{i}(CH-F_iG_{r+1-i})(-1)^{j-i}\prod_{l=1}^{j-i} \dfrac{CH -F_{i + l} G_{r+1-i-l}}{CH +F_{i + l-1} G_{r+1-i-l}}\\
&=-F_{i}\{F_j,G_k\}.
\end{align*}
Next, consider $\kappa(i,k) = p$, that is, $k=r+p+1-i$. We have
\[
\{F_j,G_k\}=
(-1)^{j+p-i} (CH -F_{i-p} G_{r+p+1-i} ) \prod_{l=1}^{p} \dfrac{CH -F_{i-p+l} G_{r+p+1-i-l} }{CH +F_{i-p+l-1} G_{r+p+1-i-l}} Q,
\]
with $Q=\prod_{l=p+1}^{j+p-i} \dfrac{CH -F_{i-p+l} G_{r+p+1-i-l} }{CH +F_{i-p+l-1} G_{r+p+1-i-l}}$.
The product $Q$ depends on $G_{r+p+1-i-l}$ with $l>p$, which implies that $Q$ commutes with $F_i$. Hence, to show that $\{F_i,\{F_j,G_k\}\}=-F_{i}\{F_j,G_k\}$, we need to prove that $\{F_i,K_p\}=-F_iK_p$ for
\[
K_p=(CH -F_{i-p} G_{r+p+1-i} ) \prod_{l=1}^{p} \dfrac{CH -F_{i-p+l} G_{r+p+1-i-l} }{CH +F_{i-p+l-1} G_{r+p+1-i-l}},
\]
for $p=1,2,\ldots$. The derivative of $\dfrac{CH -F_{i-p + l} G}{CH +F_{i-p + l-1} G}$ with respect to $G$ is $\dfrac{-CH(F_{i-p+l-1}+F_{i-p + l})}{(CH +F_{i-p + l-1} G)^2}$. The calculation for $p=1$ is
\begin{align*}
\{F_i,K_1\}&=\{F_i,(CH -F_{i-1} G_{r+2-i} ) \dfrac{CH -F_{i} G_{r+1-i} }{CH +F_{i-1} G_{r+1-i}}\}\\
&=-F_{i-1}\dfrac{CH -F_{i} G_{r+1-i} }{CH +F_{i-1} G_{r+1-i}}\{F_i,G_{r+2-i}\}
-(CH -F_{i-1} G_{r+2-i} )\dfrac{CH(F_{i-1}+F_{i})}{(CH +F_{i-1} G_{r+1-i})^2}\{F_i,G_{r+1-i}\}\\
&=F_{i-1}\dfrac{(CH -F_{i} G_{r+1-i})^2(CH -F_{i-1} G_{r+2-i} )}{(CH +F_{i-1} G_{r+1-i})^2}
-\dfrac{(CH -F_{i-1} G_{r+2-i} )CH(F_{i-1}+F_{i})(CH -F_{i} G_{r+1-i})}{(CH +F_{i-1} G_{r+1-i})^2}\\
&=\dfrac{F_{i-1}(CH -F_{i} G_{r+1-i})-CH(F_{i-1}+F_{i})}{(CH +F_{i-1} G_{r+1-i})}K_1\\
&=-F_iK_1.
\end{align*}
We are now ready to deal with general $p$. Define
$
K^0_p=CH -F_{i-p} G_{r+p+1-i}$, and
\[
K^l_p= \dfrac{CH -F_{i-p+l} G_{r+p+1-i-l} }{CH +F_{i-p+l-1} G_{r+p+1-i-l}}
\]
for $l=1,2,\ldots,p$, so that $K_p=\prod_{l=0}^{p} K^l_p$. Each factor $K^l_p$ depends on only one $G$-variable, namely $G_{r+p+1-i-l}$, we have $\{F_i,K^l_p\}=C^lK^l_p$ with
\[
C^0=(-1)^{p+1}F_{i-p}\prod_{l=1}^p K^l_p,
\]
and
\[
C^l=(-1)^{p+l+1}\dfrac{CH(F_{i-p+l-1}+F_{i-p + l})}{CH +F_{i-p + l-1} G_{r+p+1-i-l}}\prod_{m=l+1}^p K^m_p,
\]
for $l=1,2,\ldots,p$, and hence $\{F_i,K_p\}=\sum_{l=0}^p C^lK_p$. We claim that, for $q\leq p$,
\begin{equation}\label{lq}
\sum_{l=0}^q C^l=(-1)^{p+q+1}F_{i-p+q}\prod_{m=q+1}^p K^m_p, 
\end{equation}
which we prove by induction, as follows. By substituting $q=0$ in the right hand side of \eqref{lq} we find $C^0$ and, for $q<p$, we have, using \eqref{lq}
\begin{align*}
\sum_{l=0}^{q+1} C^l&= \sum_{l=0}^{q} C^l+ C^{q+1}\\
&=(-1)^{p+q+1}F_{i-p+q}\left(\prod_{m=q+1}^p K^m_p\right) + 
(-1)^{p+q}\dfrac{CH(F_{i-p+q}+F_{i-p + q+1})}{CH +F_{i-p + q} G_{r+p-i-q}}\prod_{m=q+2}^p K^m_p \\
&=(-1)^{p+q}\left(\dfrac{-F_{i-p+q}(CH -F_{i-p + q+1} G_{r+p-i-q})+CH(F_{i-p+q}+F_{i-p + q+1})}{CH +F_{i-p + q} G_{r+p-i-q}}\right)\prod_{m=q+2}^p K^m_p\\
&=(-1)^{p+q+2}F_{i-p + q+1}\prod_{m=q+2}^p K^m_p,
  \end{align*}
which equals the right hand side of \eqref{lq} with $q$ replaced by $q+1$. But now, taking $q=p$ we find $\sum_{l=0}^p C^l=-F_i$, which is what was to be shown.
\item Lastly, let $\kappa(j,k) \geq 0$ and $i\geq j$. We define $\kappa(j,k) = p$, so $k=r+p+1-j$. We have
\[
\{F_j,G_k\}=
(-1)^{p} (CH -F_{j-p} G_{r+p+1-j} ) \prod_{l=1}^{p} \dfrac{CH -F_{j-p+l} G_{r+p+1-j-l} }{CH +F_{j-p+l-1} G_{r+p+1-j-l}}.
\]
Let $i=j+q$. Then
\begin{align*}
\{F_i,G_k\}&=
(-1)^{p+q} (CH -F_{j-p} G_{r+p+1-j} ) \prod_{l=1}^{p+q} \dfrac{CH -F_{j-p+l} G_{r+p+1-j-l} }{CH +F_{j-p+l-1} G_{r+p+1-j-l}}\\
&=\{F_j,G_k\}Z,\qquad Z=(-1)^{q}\prod_{l=1}^{q} \dfrac{CH -F_{j+l} G_{r+1-j-l} }{CH +F_{j+l-1} G_{r+1-j-l}}.
\end{align*}
If we replace $i$ by $j$ in the definitions of $K^l_p$, $k=0,1,\ldots,p$, then we can write $\{F_j,G_k\}=\prod_{l=0}^{p} K^l_p$. Each factor $K^l_p$ depends on only one $G$-variable, namely $G_{r+p+1-j-l}$, we again have $\{F_i,K^l_p\}=C^lK^l_p$, this time with
\[
C^0=(-1)^{p+1}F_{j-p}\left(\prod_{l=1}^p K^l_p \right) Z,
\]
and
\[
C^l=(-1)^{p+l+1}\dfrac{CH(F_{j-p+l-1}+F_{j-p + l})}{CH +F_{j-p + l-1} G_{r+p+1-j-l}}\left(\prod_{m=l+1}^p K^m_p \right) Z,
\]
for $l=1,2,\ldots,p$, and hence $\{F_i,\{F_j,G_k\}\}=\sum_{l=0}^p C^l\{F_j,G_k\}=-F_j Z \{F_j,G_k\}=-F_j\{F_i,G_k\}$. 
\end{itemize}

\end{document}